\documentclass[a4paper,UKenglish,cleveref, autoref, thm-restate]{lipics-v2021}

\hideLIPIcs  

\title{Computing Vertex and Edge Connectivity of Graphs Embedded with Crossings\footnote{A preliminary version of this paper appeared in ESA 2024 \cite{esa_paper}.}} 

\author{Therese Biedl}{David R. Cheriton School of Computer Science, University of Waterloo, Canada}{biedl@uwaterloo.ca}{https://orcid.org/0000-0002-9003-3783}{}

\author{Prosenjit Bose}{School of Computer Science, Carleton University, Ottawa, Canada}{jit@scs.carleton.ca}{https://orcid.org/0000-0002-8906-0573}{}

\author{Karthik Murali}{School of Computer Science, Carleton University, Ottawa, Canada}{KarthikMurali@cmail.carleton.ca}{https://orcid.org/0000-0002-1825-0097}{}

\authorrunning{T.Biedl, P. Bose and K. Murali} 

\Copyright{Therese Biedl, Prosenjit Bose and Karthik Murali} 



\category{} 

\nolinenumbers 

\usepackage[dvipsnames]{xcolor}
\usepackage[]{todonotes}

\newcommand{\calT}{\mathcal{T}}
\newcommand{\X}{\ensuremath{\times}}

\usepackage{mathtools}
\DeclarePairedDelimiter\ceil{\lceil}{\rceil}
\DeclarePairedDelimiter\floor{\lfloor}{\rfloor}

\usepackage{framed}
\usepackage{cleveref}
\usepackage{tcolorbox}
\usepackage{tabularx}

\Crefname{claim}{Claim}{Claims}

\begin{document}
\date{}
\maketitle

\abstract
{Vertex connectivity and edge connectivity are fundamental concepts in graph theory that have been widely studied from both structural and algorithmic perspectives. The focus of this paper is on computing these two parameters for graphs embedded on the plane with crossings. For planar graphs---which can be embedded on the plane without any crossings---it has long been known that vertex and edge connectivity can be computed in linear time. Recently, the algorithm for vertex connectivity was extended from planar graphs to 1-plane graphs (where each edge is crossed at most once) without $\times$-crossings---these are crossings whose endpoints induce a matching. The key insight, for both these classes of graphs, is that any two vertices/edges of a minimum vertex/edge cut have small face-distance (distance measured by number of faces) in the embedding. In this paper, we attempt at a comprehensive generalization of this idea to a wider class of graphs embedded on the plane. Our method works for all those embedded graphs where every pair of crossing edges is connected by a path whose vertices and edges have a small face-distance from the crossing point. Important examples of such graphs include optimal 2-planar and optimal 3-planar graphs, $d$-map graphs, $d$-framed graphs, graphs with bounded crossing number, and $k$-plane graphs with bounded number of $\times$-crossings. For all these graph classes, we get a linear-time algorithm for computing vertex and edge connectivity.}


\maketitle



\newpage

\graphicspath{{Figures/}}


\section{Introduction}

The connectivity of a graph is an important measure of fault-tolerance, and mainly comes in two variants. The \textit{vertex connectivity} of a graph $G$, denoted $\kappa(G)$, is the minimum size of a set of vertices whose removal makes the graph disconnected. Such a set of vertices is called a \textit{minimum vertex cut} of $G$. An \textit{edge cut} of a graph $G$ is a set of edges whose removal makes the graph disconnected. The \textit{size of an edge cut}, for a graph without edge-weights, is the number of edges in the edge cut; for an edge-weighted graph, it is the sum of weights of all edges in the edge cut. (We always assume that edge-weighted graphs are simple because any set of parallel edges can be replaced by a single edge whose weight is the sum of weights of all edges in the set.) A \textit{minimum edge cut} of a graph, with or without edge-weights, is an edge cut of minimum possible size. For edge-unweighted graphs only, we use the term \textit{edge connectivity} to denote the size of a minimum edge cut. For both weighted and unweighted graphs, we use $\lambda(G)$ to denote the the size of a minimum edge cut. In what follows, we give a brief overview of the rich literature on algorithms for computing vertex connectivity and minimum edge cuts of general graphs.

\subsection{General Graphs}

We first discuss results concerning vertex connectivity. Let $G = (V,E)$ be a graph where $n := |V(G)|$ and $m := |E(G)|$. There exist linear (i.e. $O(m{+}n)$) time algorithms to decide whether $\kappa(G) \leq 3$  \cite{tarjan1972depth, hopcroft1973dividing}. In 1969, Kleitman \cite{kleitman1969methods} showed how to test $\kappa(G)\leq k$ in time $O(k^2nm)$. For $\kappa(G) \in O(1)$, the first $O(n^2)$ algorithm was by Nagamochi and Ibaraki \cite{DBLP:journals/algorithmica/NagamochiI92}. The fastest deterministic algorithm is by Gabow and takes time $O(m \cdot (n + \min\{k^{5/2}, kn^{3/4}\}))$ \cite{gabow2006using}. Recently, Saranurak and Yingchareonthawornchai \cite{SaranurakY22} gave an $\widehat{O}(m2^{O(k^2)})$ algorithm (we use $\widehat{O}(\cdot)$ to suppress an $n^{o(1)}$ factor), which is nearly-linear for all $k \in o(\sqrt{\log n})$. As of yet, there is no linear-time deterministic algorithm for testing whether $\kappa(G) > 3$. 

As with vertex connectivity, algorithms for minimum edge cuts have a long history. The earliest known algorithm for deciding whether $\lambda(G) \leq \lambda$ ran in $O(\lambda nm)$ time \cite{gomory1961multi, ford1956maximal}. The running time was subsequently improved to $O(\lambda n^2)$ for simple unweighted graphs \cite{podderyugin1973algorithm}, and $O(nm + n^2\log n)$ for weighted graphs \cite{nagamochi1992computing, stoer1997simple, frank1994edge}. In 1991, Gabow improved the running times to $O(m + \lambda^2n\log(n/\lambda))$ for simple unweighted graphs and $O(m+\lambda^2n\log n)$ for graphs with parallel edges \cite{gabow1991matroid}. These time bounds remained the fastest for many years until a recent breakthrough by Kawarabayashi and Thorup who gave an $O(m\log^{12}n)$ time algorithm for simple unweighted graphs \cite{Kawarabayashi}. This was subsequently improved by Henzinger, Rao and Wang to $O(m\log^2n\log \log^2n)$ for simple unweighted graphs \cite{MonikaHenzinger}, and is the state-of-art time complexity for edge connectivity. Sarunarak simplified this algorithm at the cost of a slower $\widehat{O}(m)$ running time \cite{saranurak} for simple unweighted graphs, while Li achieved the same $\widehat{O}(m)$ running time for weighted graphs \cite{JasonLi}.

\subsection{Planar Graphs and Beyond}

Algorithms for computing minimum vertex and edge cuts for planar graphs are significantly faster than for general graphs. Eppstein \cite{eppstein2002subgraph} showed that vertex and edge connectivity of a simple unweighted planar graph can be computed in $O(n)$ time. For planar graphs with at most $t$ parallel (unweighted) edges connecting any pair of vertices, Eppstein's algorithm runs in time $2^{O(t \log t)}n$. 
For weighted planar graphs, an $O(n\log^2n)$ time algorithm for computing minimum edge cuts was given by Chalermsook et al. \cite{chalermsook2004deterministic}, and later improved by Italiano et al. to $O(n\log n \log\log n)$ \cite{italiano2011improved}. The state-of-art algorithm is by \L\k{a}cki and Sankowski which runs in $O(n \log\log n)$ time \cite{lkacki2011min}. In \cite{biedl2022computing}, Biedl and Murali extended Eppstein's method for computing vertex connectivity of planar graphs to 1-plane graphs without $\times$-crossings. These are graphs which can be drawn on the plane such that each edge is crossed at most once, and for any pair of edges $\{e,e'\}$ that cross each other, there is a path of length 3 that begins and ends with $e$ and $e'$ (crossings not satisfying this property are called $\times$-crossings). 

\subsection{Our results.}

A crucial property that enabled the design of linear-time connectivity testing algorithms for planar graphs \cite{eppstein2002subgraph} and 1-plane graphs without $\X$-crossings \cite{biedl2022computing} is that the vertices of a minimum vertex cut (and also the edges of a minimum edge cut for planar graphs, in \cite{eppstein2002subgraph})
were shown to be close to each other in the embedding. Put differently, any two vertices/edges of a minimum vertex/edge cut are separated by a small number of faces in the embedding. Naturally, one is interested to know what minimum condition an embedded graph must satisfy to enable this property about minimum vertex and edge cuts. Any such condition certainly includes the whole class of planar graphs, but not all of 1-plane graphs (see \cite{biedl2022computing} for a restrictive example). Preferably, one must also be able to describe it easily and test it in polynomial time.   

In this paper, we attempt to formulate an answer through a parameter called ribbon radius. Let $G$ be a graph embedded on the plane with crossings\footnote{In some contexts, the term `embedded' may refer to graphs drawn on surfaces without crossings. However, our use of the term `embedded' refers to any drawing of a graph on the plane, including drawings that feature crossings.}. Loosely speaking, the \textit{ribbon radius} of a graph is the radius of the smallest ball which, when centered at any crossing point, contains a \textit{ribbon} at the crossing---that is, a path that starts and ends with the pair of edges involved in the crossing. Given a drawing of a graph $G$, let $G^\times$ be the \textit{planarization} of $G$ obtained by inserting a dummy vertex at each crossing point, and replacing the crossing pair of edges with four edges incident to the dummy vertex. Let $\Lambda(G)$ be the graph obtained from $G^\times$ by inserting a face-vertex inside each face of $G^\times$ and making it adjacent to vertices on the face-boundary. Our main structural result is the following: If a graph $G$ has ribbon radius $\mu := \mu(G)$, then all vertices of a minimum vertex cut and edges of a minimum edge cut lie in a subgraph of $\Lambda(G)$ with diameter $O(\mu\kappa)$ and $O(\mu\lambda)$ respectively. This structure helps us design an algorithm for testing connectivity, for which we make use of the framework developed in \cite{biedl2022computing}. The idea is to construct the auxiliary graph $\Lambda(G)$, decompose it into subgraphs of diameter $O(\mu\kappa)$ or $O(\mu\lambda)$, and then run a dynamic programming algorithm on the tree decomposition of these subgraph. From the algorithmic front, our main result is the following: Given an embedded graph $G$ (in the form of its planarization $G^\times$) and its ribbon radius $\mu(G)$, we can compute $\kappa(G)$ and $\lambda(G)$ in time $2^{O(\mu\kappa)}|V(G^\times)|$ and $2^{O(\mu\lambda)}|V(G^\times)|$ respectively. Several well-known graph classes such as graphs with bounded crossing number, $d$-map graphs, $d$-framed graphs, optimal 2-planar and optimal 3-planar graphs etc. have bounded ribbon radius. This gives us an $O(n)$-time algorithm for testing vertex and edge connectivity for these graph classes. See \Cref{table: results} for a summary of our results.

\subparagraph{Organization of the Paper} In \Cref{section: preliminaries}, we give all preliminary definitions related to graphs and graph drawings. In addition, we discuss two concepts from  \cite{biedl2022computing}---radial planarization (the graph $\Lambda(G)$) and co-separating triples---which are important building blocks for this paper. In \Cref{section: ribbon radius}, we define ribbon radius formally, and give a polynomial-time algorithm for computing it. Following this, in \Cref{sec:applications}, we show that many classes of near-planar graphs have small ribbon radius. Our main structural result appears in \Cref{sec: co-separating triples}, where we show that if a graph $G$ has a small ribbon radius, then the vertices of any minimal vertex cut are inside a small-diameter subgraph of $\Lambda(G)$. Building on this, \Cref{sec: finding min vertex and edge cut} presents an algorithm---based on the framework of \cite{biedl2022computing}---for computing vertex and edge connectivity of embedded graphs. The implications of this for near-planar graphs are discussed in \Cref{sec: algorithmic applications}. In \Cref{app :sec: large ribbon radius}, we discuss the limitations of our method by showing how one can construct $k$-plane graphs with large ribbon radius. We conclude the paper in \Cref{sec: outlook}.

\begin{table}[ht!]
\centering
\renewcommand{\arraystretch}{1.3}
\setlength{\tabcolsep}{3pt}
\begin{tabular}{
    | >{\raggedright\arraybackslash}m{0.36\linewidth}
    | >{\centering\arraybackslash}m{0.12\linewidth}
    | >{\centering\arraybackslash}m{0.12\linewidth}
    | >{\centering\arraybackslash}m{0.14\linewidth}
    | >{\raggedright\arraybackslash}m{0.16\linewidth} |
}
\hline
\textbf{Graph Class} &
\textbf{$\kappa(G)$ or $\lambda(G)$} &
\textbf{$\mu(G)$} &
\textbf{Runtime} &
\textbf{Reference} \\
\hline
$k$-plane graphs without $\times$-crossings &
$O(\sqrt{k})$ & $O(k)$ &
$2^{O(k^{1.5})}n$ &
\Cref{cor: main result k-plane,claim: k-plane no x crossings} \\
\hline
$k$-plane graphs with at most $\gamma$ $\times$-crossings &
$O(\sqrt{k})$ & $O(\gamma + k)$ &
$2^{O(\sqrt{k}(\gamma + k))}n$ &
\Cref{cor: main result k-plane,cor: bounded x crossings ribbon radius} \\
\hline
$k$-plane graphs with ribbon-length $\leq \gamma$ at each crossing &
$O(\sqrt{k})$ & $O(\gamma k)$ &
$2^{O(\gamma k^{1.5})}n$ &
\Cref{cor: main result k-plane} \\
\hline
Graphs with crossings only in faces of $\mathrm{sk}(G)$ whose boundary is a cycle of length at most $d$. This includes $d$-framed graphs and $d$-map graphs. &
$O(d)$ & $O(d^2)$ &
$2^{O(d^3)}n$ &
\Cref{cor: dframed and dmap} \\
\hline
Graphs with at most $q$ crossings &
$O(\sqrt{q})$ & $O(q)$ &
$2^{O(q^{1.5})}n$ &
\Cref{cor: q crossings per face} \\
\hline
Graphs with at most $q$ crossings in each face of $\mathrm{sk}(G)$ &
$O(\sqrt{q})$ & $O(q)$ &
$2^{O(q^{1.5})}n$ &
\Cref{cor: q crossings per face} \\
\hline
$k$-plane graphs with $\mu(G) \in O(1)$ and $k \in O(1)$. This includes optimal $k$-plane graphs for $k \in \{1,2,3\}$. &
$O(1)$ & $O(1)$ &
$O(n)$ &
\Cref{cor: kplane constant mu} \\
\hline
Any graph with $\kappa(G) \cdot \mu(G)$ or $\lambda(G) \cdot \mu(G)$ in $o(\log n)$ &
& &
$\widehat{O}(|V(G^\times)|)$ &
\Cref{theorem: vertex connectivity} \\
\hline
Graphs with $o(\log n)$ crossings &
$O(1)$ & $o(\log n)$ &
$\widehat{O}(n)$ &
\Cref{prop: crossing lemma} \\
\hline
\end{tabular}
\vspace{1em}
\caption{A table showing a summary of results for various classes of near-planar graphs. (We assume intersection-simple drawings for many of these graph classes.) We use $\widehat{O}(f(n))$ to denote $O(f(n)\cdot n^{o(1)})$.}
\label{table: results}
\end{table}


\section{Preliminaries and Background}\label{section: preliminaries}

Let $G = (V,E)$ be a graph. A \textit{$(u,v)$-path} in $G$ is a path $P$ with end vertices $u$ and $v$. For any two vertices $u', v' \in V(P)$, we use $P(u'\dots v')$ to denote the sub-path of $P$ with end vertices $u'$ and $v'$. The length of $P$, denoted by $|P|$, is the number of edges on $P$. Graph $G$ is said to be \textit{connected} if, for every pair of vertices $u,v \in V(G)$, there is a $(u,v)$-path in $G$; otherwise $G$ is said to be \textit{disconnected}. A set of vertices $S \subseteq V(G)$ is a \textit{vertex cut} if $G \setminus S$ is disconnected, or is an isolated vertex. A vertex cut of minimum size is called a \textit{minimum vertex cut} of $G$; its size is the \textit{vertex connectivity} of $G$, denoted by $\kappa(G)$. Likewise, a set of edges $T$ is called an \textit{edge cut} of $G$ if $G \setminus T$ is disconnected. A \textit{minimum edge cut} of a graph $G$ without edge-weights is an edge cut of minimum size; the size is called the edge connectivity of $G$, and is denoted by $\lambda(G)$. All graphs in this paper are edge-unweighted, but need not be simple---they may have loops or parallel edges. 
For all graphs, the minimum degree $\delta(G)$ is an upper bound for both $\kappa(G)$ and $\lambda(G)$; in fact $\kappa(G) \leq \lambda(G) \leq \delta(G)$. 

For a vertex cut $S$ or an edge cut $T$, we call the connected components of $G \setminus S$ and $G \setminus T$ the \textit{flaps} of $S$ and $T$, respectively. 
A set of vertices $S$ \textit{separates} two non-empty sets $A$ and $B$ if no flap of $G\setminus S$ contains vertices of both $A$ and $B$. A set $S$ is a \textit{minimal vertex cut} if no subset of $S$ is a vertex cut. Likewise, a set $T$ is a \textit{minimal edge cut} if no subset of $T$ is an edge cut. Set $S$ is a minimal vertex cut if and only if each vertex of $S$ has a neighbour in every flap of $S$. If $T$ is a minimal edge cut of $G$, then there are exactly two flaps, and every edge of $T$ connects a vertex of one flap with a vertex of the other flap. For any set $Q \subseteq V(G)$, the notation $N_G[Q]$ denotes the \textit{closed neighbourhood} of $Q$ in the graph $G$, which is the set of all vertices that are either in $Q$ or adjacent to some vertex of $Q$. 


\subsection{Graph Drawings}\label{subsection: embedded graphs}

(We follow the conventions of \cite{schaefer2012graph} to define a graph drawing.) A \textit{drawing} of a graph $G = (V,E)$ is a mapping of $V(G)$ to distinct points on the plane, and each edge $e \in E(G)$ is a homeomorphic mapping of the interval $[0,1]$ on the plane, where $e(0)$ and $e(1)$ are the endpoints of the edge and $e(0,1)$ does not contain any vertices. A \emph{crossing point} of two edges $\{e_1,e_2\}$ is an intersection point $c = e_1(s) = e_2(t)$ that is not a vertex; i.e. $s,t\in (0,1)$. We require that no other edge passes through $c$, and that $e_1$ and $e_2$ cross transversally at $c$ (and not touch each other tangentially). A drawing of a graph is \textit{intersection-simple} \cite{schaefer2012graph} if any two edge-curves intersect at most once, either at a common incident vertex, or at a crossing point.  
We do not assume our graph drawings to be intersection-simple, unless stated otherwise. The \textit{crossing number} of a graph $G$, denoted by $\mathrm{cr}(G)$, is the minimum number of crossing points over all drawings of $G$.

\subparagraph{Crossings} If $e_1$ and $e_2$ are a crossing pair of edges, then the \emph{endpoints} of the crossing are the endpoints of $e_1$ and $e_2$. Two endpoints of a crossing are said to be \textit{consecutive} if one belongs to $e_1$ and the other to $e_2$; else they are said to be  \textit{opposite}. Since we allow drawings that are not intersection-simple, it is possible for two consecutive endpoints to be the same vertex. However, this will not cause any issues, as the challenging cases---as we will see---occur when consecutive endpoints are, in some sense, far apart. A drawing $D$ of a graph $G$ is \textit{planar} if it has no crossing points. The maximal connected regions of $\mathbb{R}^2 \setminus D$ are called \textit{faces} of the planar drawing. The faces can be described abstractly via a \emph{planar rotation system}, which specifies a clockwise order of edges incident with each vertex. This defines the \emph{face boundaries}, which are closed walks where each successive edge comes after the previous edge in the rotation system at the common endpoint. The \textit{skeleton} of a graph $G$ is the plane subgraph $\mathrm{sk}(G)$ that has the same vertex set as $G$, but includes only the uncrossed edges of $G$. The skeleton may or may not be connected; it may even be a graph without any edges. Every crossed edge of $G$ and every crossing point lies within some face of $\mathrm{sk}(G)$; we loosely say that the face \textit{contains} the edge or crossing point.

\subparagraph{Planarization} 
The \textit{planarization} of a drawing of a graph $G$ is the plane graph $G^\times$ obtained by inserting a \textit{dummy vertex} at each crossing point and replacing the crossing pair of edges with four edges incident to the dummy vertex. In this paper, we will use the terms dummy vertex and crossing point synonymously.  An \textit{embedding} of a graph $G$ is any drawing of $G$ expressed via a planar rotation system of $G^\times$. For any subgraph $H \subseteq G$, we use $H^\times$ to denote the sub-drawing induced by $H$ in $G^\times$. For instance, if $e = (u,v)$ is an edge of $G$, then $e^\times$ is the $(u,v)$-path in $G^\times$ whose internal vertices (if any) are the crossing points on $e$. The edges of $e^\times$ are called the \emph{edge-segments} of $e$, while $e$ is the \emph{parent-edge} for its edge-segments. For each dummy vertex $d$ on $e$, the sub-paths $e^\times(d\dots u)$ and $e^\times(d \dots v)$ are called \textit{part-edges}.

\subsection{Radial Planarization and Co-Separating Triple}

The concepts of radial planarization and co-separating triple were introduced in \cite{biedl2022computing} by an abstraction of the methods used for testing vertex connectivity of planar graphs \cite{eppstein2002subgraph}. These concepts are of vital importance to the present paper, so we re-state these definitions here.

\begin{definition}[Radial Planarization]\label{def: radial planarisation}
    For any drawing of a graph $G$ on the plane, the radial planarization $\Lambda(G)$ is the graph obtained by inserting a face-vertex inside each face of $G^\times$ and making it adjacent to all vertices of $G^\times$ that bound the face.
\end{definition}

(The definition in \cite{biedl2022computing} is slightly more stringent as every incidence of a face with a vertex adds an edge to the corresponding face-vertex, potentially creating parallel edges in $R(G)$. For our purposes, the above definition suffices.) The subgraph of $\Lambda(G)$ induced by all edges incident to face-vertices is called the \textit{radial graph} of $G$, denoted by $R(G)$. Note that $R(G)$ is a simple bipartite graph with bipartition $(A,B)$, where $A = V(G^\times)$ and $B$ is the set of all face-vertices. The \textit{face-distance} between any two vertices $u,v \in \Lambda(G)$, denoted by $d_F(u,v)$, is the number of face-vertices on a shortest path between $u$ and $v$ in $R(G)$. For any subset of vertices $S$ in $\Lambda(G)$, we write $d_F(u,S) = \min_{v \in S}d_F(u,v)$.

\begin{definition}[Co-separating Triple and Nucleus]\label{def: co-sep}
Let $G$ be an embedded graph and $\Lambda$ be a graph where $V(\Lambda) \subseteq V(\Lambda(G))$. A partition of the vertices of $\Lambda$ into three sets $(A,X,B)$ is called a \emph{co-separating triple of $(\Lambda, G)$} if it satisfies the following conditions: 
\begin{enumerate}[(C1)]
    \item Each of $A,X,B$ contains at least one vertex of $G$. \label{C1}
    
    \item There is no edge of $\Lambda$ with one endpoint in $A$ and the other endpoint in $B$.\label{C2}
    
    \item For every edge $e \in G \setminus (X \cap V(G))$, all vertices of $e^\times \cap \Lambda$ belong entirely to $A \cup X$, or belong entirely to $B \cup X$. \label{C3}
\end{enumerate}    
    The \emph{nucleus} of $(A,X,B)$ is the set $N_G[X \cap V(G)]$. The maximum distance in $\Lambda$ between any two vertices of the nucleus is called the \emph{nuclear-diameter}.
\end{definition}


\section{Defining and Computing Ribbon Radius}
\label{section: ribbon radius}

Our algorithm to compute vertex and edge connectivity uses a parameter that we call ribbon radius. For any crossing pair of edges $\{e_1,e_2\}$, a \emph{ribbon} is a path in $G$ that begins with $e_1$ and ends with $e_2$. Informally, the ribbon radius of a graph is the radius of a smallest ball in $\Lambda(G)$ that, when centered at any crossing point, contains a ribbon at that crossing within the ball. To formalize this, we need some definitions first. 

For any set $S \subseteq V(\Lambda(G))$ and integer $r \geq 0$, define $\Lambda(S,r)$ 
as the subgraph of $\Lambda(G)$ induced by all vertices $v$ such that $d_F(v,S) \leq r$. If $S = \{c\}$ is a singleton, we simply write $\Lambda(c,r)$ instead of $\Lambda(\{c\},r)$. Let $\mathcal{B}(S,r)$ 
be the restriction of $\Lambda(S,r)$ to $G^\times$ obtained by deleting all face-vertices. 
We say that a vertex/edge of $G^\times$ is on the \emph{boundary} of $\mathcal{B}(S,r)$ if it has an incident face-vertex of face-distance $r$ and an incident face-vertex of face-distance $r+1$.   (Note that face-distance of incident face-vertices can differ by at most one.)   Let $\mathcal{Z}(S,r)$ be the graph formed by the vertices and edges on the boundary of $\mathcal{B}(S,r)$. Being a subgraph of the planar graph $G^\times$, we have that $\mathcal{Z}(S,r)$ is also planar. Each face of $\mathcal{Z}(S, r)$ corresponds to a union of faces of $G^\times$, and can be assigned one of two colours depending on whether the corresponding faces in $G^\times$ have face-distance at most $r$ or at least $r+1$. Those with face-distance at least $r+1$ are referred to as \emph{holes} of $\Lambda(S, r)$. Since the two faces incident to each boundary edge have different colours, by walking along the boundary of a face of $\mathcal{Z}(S,r)$, we can possibly encounter vertices repeatedly, but edges at most once; we call such a closed walk a \emph{circuit}. Therefore, all faces of $\mathcal{Z}(s,r)$ are bounded by circuits. 
%

Let $v$ be a vertex of $G^\times$. If $v$ is a dummy vertex, then define $\mu(v)$ as the smallest integer such that $B(v, \mu(v))$ contains $Q_v \times$ as a subgraph, for some ribbon $Q_v$ at $v$.
Put differently, $\mathcal{B}(v,\mu(v))$ must contain the crossing point $v$, all four part-edges at $v$, and for some path in $G$ that connects two consecutive endpoints of the crossing at $v$, all edge-segments and dummy vertices of all edges on that path. If $v$ is a vertex of $G$, then define $\mu(v):=0$. The value $\mu(v)$ is called the \textit{ribbon radius at $v$}. The \textit{ribbon radius of $G$}, denoted by $\mu(G)$, is then defined as $\mu(G) := 1 + \max_{v \in V(G^\times)} \mu(v)$\footnote{The `1+' in the definition may seem unusual, but will be crucial later in \Cref{theorem: vertex connectivity}.}. In \Cref{obs: face distance between dummy-vertices of an edge}, we show that the face-distance between any two vertices on $e^\times$, for some $e \in E(G)$, is at most the ribbon-radius $\mu(G)$.

\begin{observation}\label{obs: face distance between dummy-vertices of an edge}
Let $e = (u,v)$ be an edge of $G$. Then $d_F(u,v) \leq \mu(G)$, and for any dummy vertex $w \in e^\times$, $\max\{d_F(u,w), d_F(v,w)\} \allowbreak \leq \mu(w)$.
\end{observation}

\begin{proof}
For any dummy vertex $w \in e^\times$, the graph $\mathcal{B}(w,\mu(w))$ contains all edges of $e^\times$. Hence, $\max\{d_F(u,w), d_F(v,w)\} \allowbreak  \leq \mu(w)$. We now show that $d_F(u,v) \leq \mu(G)$. If $e$ is an uncrossed edge, then $d_F(u,v)=1$, and by $\mu(G)\geq 1$, the result holds. If $e$ is crossed, then choose $w$ to be the first dummy vertex in the direction $u$ to $v$. As $d_F(u,w) = 1$ and $d_F(w,v) \leq \mu(w) \leq \mu(G)-1$, we have $d_F(u,v) \leq \mu(G)$.
\end{proof}

If $Q$ is a ribbon in a graph $G$, then $Q^\times$ is a walk in $G^\times$ with some properties---for example, the walk does not `turn' or `reverse' at dummy vertices. We find it useful to formalize these notions for arbitrary walks in $G^\times$. Consider a walk $W$ in $G^\times$ that contains a dummy vertex $d$ as an interior node; say the incident edge-segments at $d$ are $e_1$ and $e_2$.  We say that $W$ \emph{reverses at $d$} if $e_1=e_2$, \emph{goes straight at $d$} if $e_1\neq e_2$ but have the same parent-edge, and \emph{makes a turn} otherwise. We say that $W$ is a \textit{$G$-respecting} if it begins and ends at vertices of $G$ and goes straight at all crossing points (thereby tracing a walk in $G$).

Define a \textit{ribbon-loop} in $G^\times$ as a closed walk that: (a) begins and ends at a crossing point $v$; (b) its first and last edge-segments have different parent-edges; (c) it makes no turns, but possibly reverses. Notice that if $Q_v$ is a ribbon at $v$, then the sub-walk of $Q^\times_v$ between the two occurrences of $v$ satisfies all three conditions, and so is a ribbon-loop. In \Cref{claim: ribbon radius in terms of ribbon length}, we show a useful way to bound the ribbon radius at a crossing point using ribbon-loops.

\begin{observation}\label{claim: ribbon radius in terms of ribbon length}
Let $e_1 = (u,v)$ and $e_2 = (w,x)$ be a pair of edges crossing at point $c$. Let $Q_c$ be a ribbon at $c$ with $u$ and $w$ as its end vertices. Then $\mu(c) \leq \max \Big \{ \left|e_1^\times(c\dots u)\right|, \left |e_2^\times(c\dots w) \right|, \allowbreak 
 \frac{\left|E(Q_c^\times)\right|}{2} - 1 \Big \}$.
\end{observation}

\begin{proof}
To prove this, we use the fact that the face-distance between any two vertices of $G^\times$ is at most the length of a shortest path between them in $G^\times$ (since we can use face-vertices near every edge of $G^\times$). Any vertex $t \in Q_c^\times$ belongs to either $e_1^\times(c \dots u) \cup e_2^\times(c \dots w)$ or to the ribbon-loop $Q_c^\times \setminus \allowbreak E\left(e_1^\times(c \dots u) \cup e_2^\times(c \dots w)\right)$. In the former case, $d_F(c,t) \leq \max\{|e_1^\times(c\dots u)|, |e_2^\times(c \dots w)|\}$. In the latter case, $t$ is a part of a closed walk with at most $|E(Q_c^\times)|-2$ edges. Hence $d_F(c,t) \leq \frac{\left|E(Q_c^\times)\right|-2}{2} =  \frac{\left|E(Q_c^\times)\right|}{2} - 1$.     
\end{proof}

The definition of ribbon radius leads to a straight-forward polynomial-time algorithm for computing it. Given an embedded graph $G$ (in the form of a planar rotation system of $G^\times$), we can construct $\Lambda(G)$, and use breadth-first search (BFS) at each crossing point $v$ to determine $\mu(v)$. Once this is done, $\mu(G)$ is simply one plus the maximum value of $\mu(v)$ over all crossing points $c$. In \Cref{prop: compute ribbon radius}, we explain this simple algorithm in more detail, and optimize it to be a little more time-efficient.

\begin{theorem}\label{prop: compute ribbon radius}
Given any embedded graph $G$ with $n$ vertices and $q$ crossing points, one can compute $\mu(G)$ in time $O(q(n+q)\log \mu(G))$. 
\end{theorem}

\begin{proof}
We first compute the graphs $\Lambda(G)$ and $R(G)$;  this can be done in $O(n+q)$ time since we are given $G^\times$ and $|V(G^\times)| = n+q$. The main idea is now to do a galloping search for the correct value of $\max_{v \in V(G^\times)} \mu(v) = \mu(G)-1$. We try values of $\mu = 1,2,4,8,\dots$ until we succeed, say at value $\mu'$, and then do a binary search for the correct value in the interval between $\mu'/2$ and $\mu'$. Therefore, we try for $O(\log(\mu(G))$ values in total. For each of these values, we must test whether $\mu(v) \leq \mu$ for all the $q$ dummy-vertices $v$. If we show that one can test whether $\mu(v) \leq \mu$ in time $O(n+q)$, then this will imply that we spend $O(q(n+q)\log \mu(G))$ time in total, as required.

To test whether $\mu(v) \leq \mu$, perform a BFS starting at $v$ in $R(G)$, restricting the search only to vertices at distance within $2\mu$ from $v$; this determines exactly the vertices of $\Lambda(v,\mu)$ from which we can obtain $\mathcal{B}(v,\mu)$. Then, test whether all four part-edges at $v$ are in $\mathcal{B}(v,\mu)$; if not then $\mu(v)>\mu$ by \Cref{obs: face distance between dummy-vertices of an edge} and we can stop and return `no'. Otherwise, we modify $\mathcal{B}(v,\mu)$ into a subgraph $\mathcal{B}_G(v,\mu)$ of $G$: keep all vertices of $V(G) \cap V(\mathcal{B}(v,\mu))$, and insert all those edges $e \in E(G)$ for which $e^\times \subseteq \mathcal{B}(v,\mu)$. Then $\mathcal{B}(v,\mu)$ contains a ribbon if and only if two consecutive endpoints of the crossing at $v$ are in the same connected component of $\mathcal{B}_G(v,\mu)$; this can easily be tested in time $O(|\mathcal{B}_G(v,\mu)|) \subseteq O(|\mathcal{B}(v,\mu)|) \subseteq O(n+q)$.
\end{proof}


\section{Graphs with Small Ribbon Radius}
\label{sec:applications}

We now give examples of several classes of near-planar graphs that have small ribbon radius. `Near-planar graphs' is an informal term for graphs that are close to being planar; these have received a lot of attention in recent years (see \cite{didimo2019survey, hong2020beyond} for surveys in this area). Of special interest to us is the class of \textit{$k$-plane graphs}, which are graphs that are embedded on the plane with each edge crossed at most $k$ times.

\begin{figure}
    \centering
    \includegraphics[scale = 0.75]{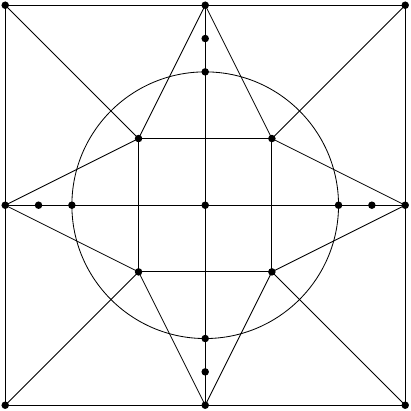}
    \caption{A graph drawing where all crossings are $\times$-crossings}
    \label{fig: xcrossings}
\end{figure}

\begin{proposition}\label{cor: small ribbon length}
If $G$ is a $k$-plane graph such that there exists a ribbon of length at most $\gamma$ at every crossing of $G$, then $\mu(G) \leq \gamma (k+1)/2$.
\end{proposition}

\begin{proof}
Let $e_1 = (u,v)$ and $e_2 = (w,x)$ be a pair of edges that cross at a point $c$. Let $Q_c$ be a ribbon at the crossing of length at most $\gamma$; up to renaming the endpoints of the crossing, assume that the end vertices of the ribbon are $u$ and $w$.  Since the graph is $k$-plane, every edge $e$ has at most $k$ crossings, so $e^\times$ has at most $k+1$ edge segments.  Therefore $|E(Q_c^\times)|\leq  \gamma(k+1)$, and  $|e_1^\times(c \dots u)| \leq k$ and $|e_2^\times(c \dots w)| \leq k$.  
By \Cref{claim: ribbon radius in terms of ribbon length}, we have $\mu(c) \leq \max\{|e_1^\times(c \dots u)|,|e_2^\times(c \dots v)|,|E(Q_c^\times)|/2-1\} \leq \max\{k, \gamma (k+1)/2 - 1\} = \gamma (k+1)/2 - 1$, since $\gamma \geq 3$. Therefore $\mu(G) \leq \gamma (k+1)/2$.
\end{proof}

An $\times$-crossing of a 1-plane graph, as defined in \cite{biedl2022computing}, is a crossing whose endpoints induce a matching. 
We can extend this definition to general embedded graphs without the requirement of 1-planarity: an \textit{$\times$-crossing} of any embedded graph is one whose endpoints induce a matching (see \Cref{fig: xcrossings} for an example). An embedded graph without $\times$-crossings has a ribbon of length at most 3 at every crossing. (The ribbon length is less than 3 if the endpoints of the two crossing edges are not disjoint.) Therefore, we get the following corollary of \Cref{cor: small ribbon length}. 

\begin{corollary}\label{claim: k-plane no x crossings}
If $G$ is a $k$-plane graph without $\times$-crossings, then $\mu(G) \leq 3(k+1)/2$. 
\end{corollary}

In light of \Cref{claim: k-plane no x crossings}, one may ask whether the ribbon radius of a graph with bounded number of $\times$-crossings is also bounded. We answer this more generally by showing that graphs where all but a few crossings have small ribbon radii have small value of $\mu(G)$.

\begin{theorem}\label{proposition: bounds on ribbon radius}
Let $G$ be a connected embedded graph, and $\alpha,\gamma$ be two constants such that there are at most $\gamma$ crossing points $c$ with $\mu(c)>\alpha$. Then $\mu(G) \leq \gamma + \alpha+1$. 
\end{theorem}

\begin{proof}
It suffices to show that $\mu(c)\leq \gamma+\alpha$ for every crossing point $c$. This holds trivially if $\mu(c)\leq \alpha$, so we may assume that $\mu(c) > \alpha$. For $r=0,1,2,\dots,\gamma$, consider the boundaries $\mathcal{Z}(c,r)$ of 
the graphs $\mathcal{B}(c,r)$. 
By definition, any vertex of $\mathcal{Z}(c,r)$ has incident face-vertices of $G^\times$ with face-distance $r$ and $r{+}1$, and no other face-distances are possible.
Therefore $\mathcal{Z}(c,0),\dots,\mathcal{Z}(c,\gamma)$ form a set of $\gamma{+}1$ vertex-disjoint graphs. As there are at most $\gamma-1$ crossing points other than $c$ whose ribbon radius is more than $\alpha$, there must exist some $1 \leq r_c\leq \gamma$ such that $\mathcal{Z}(c,r_c)$ contains none of them, i.e., $\mu(z)\leq \alpha$ for all $z\in \mathcal{Z}(c,r_c)$. Having fixed such a value $r_c$, we will show how to construct a ribbon at $c$ within $\mathcal{B}(c,r_c+\alpha)$. Let $\{e_1,e_2\}$ be the two edges crossing at $c$.

\begin{claim}\label{claim: planarized crossing within ball r_c + alpha}
$\{e_1^\times,e_2^\times\} \subseteq \mathcal{B}(c,r_c+\alpha)$.
\end{claim}

\begin{proof}
Assume for contradiction
that $e_1^\times$ is not a subgraph of $\mathcal{B}(c,r_c+\alpha)$. Since $r_c > 0$, point $c$ is within $\mathcal{B}(c,r_c)$, and the path $e_1^\times$ must leave $\mathcal{B}(c,r_c+\alpha)$ somewhere; hence there exists a vertex $z \in e_1^\times \cap \mathcal{Z}(c,r_c)$ and $z$ must be a dummy-vertex since it is in the interior of $e_1^\times$. By choice of $r_c$, we have $\mu(z)\leq \alpha$, so there exists a ribbon $Q_z$ at $z$ such that $Q_z^\times \subseteq \mathcal{B}(z,\alpha)$.  We have $e_1^\times \subseteq Q_z^\times$ since $e_1$ crosses some other edge $e'$ at $z$, and the ribbon $Q_z$ includes edge $e_1$ (Figure~\ref{fig:gamma_plus_alpha:ex}). So, $e_1^\times \subseteq \mathcal{B}(z,\alpha)\subseteq \mathcal{B}(c,r_c+\alpha)$.
\end{proof}

\begin{figure}
     \centering
     \begin{subfigure}[b]{0.3\textwidth}
         \centering
         \includegraphics[scale = 0.75,page=1]{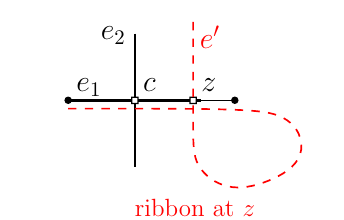}
         \caption{$e_1^\times$ is in $\mathcal{B}(c, r_c + \alpha)$}
         \label{fig:gamma_plus_alpha:ex}
     \end{subfigure}
     \hfill
     \begin{subfigure}[b]{0.3\textwidth}
         \centering
         \includegraphics[scale = 0.75,page=2]{Figures/gamma_plus_alpha.pdf}
         \caption{Replacing $W$ with $W'$.}
         \label{fig:gamma_plus_alpha:outside}
     \end{subfigure}
     \hfill
     \begin{subfigure}[b]{0.33\textwidth}
         \centering
         \includegraphics[scale = 0.75,page=3]{Figures/gamma_plus_alpha.pdf}
         \caption{Inserting a ribbon-loop $W_z$.}
         \label{fig:gamma_plus_alpha:turn}
     \end{subfigure}
\end{figure}

\begin{claim}\label{lemma: ribbon-loop}
    There is a ribbon-loop $Q$ at $c$ within $\mathcal{B}(c, r_c + \alpha)$.
\end{claim}

\begin{proof}
Let $W$ be any ribbon-loop at $c$ (this exists because $G$ is connected, and any ribbon defines a ribbon-loop). We will convert $W$ into a walk $W'$ at $c$ (not necessarily a ribbon-loop) that strictly stays within $\mathcal{B}(c,r_c)$. Since $r_c>0$, vertex $c$ is strictly inside one face $f_c$ of the boundary-graph $\mathcal{Z}(c,r_c)$. Assume that $W$ goes outside face $f_c$. Then there is circuit $Z \subseteq \mathcal{Z}(c,r_c)$ such that $c$ is inside and some other vertex of $W$ is outside $Z$. But to return to $c$, walk $W$ must return to $Z$ at some later point. We can therefore re-route $W$ by walking along circuit $Z$ (Figure~\ref{fig:gamma_plus_alpha:outside}). By repeating at all instances where $W$ goes outside $f_c$, we get the walk $W' \subseteq \mathcal{B}(c,r_c)$. 

Walk $W'$ begins and ends with the same edge-segments as $W$, but it may have turns at dummy-vertices $z \in \mathcal{Z}(c,r_c)$ due to the re-routing of $W$. However, since $\mu(z)\leq \alpha$, there is a ribbon-loop $W_z$ at $z$ that stays within $\mathcal{B}(z,\alpha)\subseteq \mathcal{B}(c,r_c+\alpha)$. By inserting $W_z$ at $z$, we can remove the turn at $z$ without adding any new turns (\Cref{fig:gamma_plus_alpha:turn}). By repeating at all such instances, we obtain a ribbon-loop $Q \subseteq \mathcal{B}(c,r_c+\alpha)$.
\end{proof}

From \Cref{claim: planarized crossing within ball r_c + alpha,lemma: ribbon-loop}, we can obtain a ribbon at $c$ as follows. If the ribbon-loop $Q$ (from \Cref{lemma: ribbon-loop}) reverses at some dummy vertex, then we simply omit the twice-visited edge-segment. Since $Q$ now does not turn and goes straight at dummy vertices, it must begin and end with edge-segments of $e_1$ and $e_2$. To $Q$, we attach the part-edges of $e_1$ and $e_2$ that are possibly not visited (these are within $\mathcal{B}(c,r_c+\alpha)$ by \Cref{claim: planarized crossing within ball r_c + alpha}). This gives a $G$-respecting walk within $\mathcal{B}(c,r_c+\alpha)$ from which a ribbon can be obtained.
\end{proof}

The corollaries below follow easily from \Cref{proposition: bounds on ribbon radius,claim: k-plane no x crossings}.

\begin{corollary}\label{cor: bounded x crossings ribbon radius}
If a $k$-plane graph $G$ has at most $\gamma$ $\X$-crossings, then $\mu(G) \leq \gamma + \frac{3(k+1)}{2}$. 
\end{corollary}

\begin{corollary}\label{cor: bounded crossings}
If a graph $G$ can be embedded with at most $q$ crossing points, then  $\mu(G) \leq q + 1$.
\end{corollary}

In fact, one can generalize \Cref{cor: bounded crossings} to all graphs that contain few crossings in each face of its skeleton $\mathrm{sk}(G)$ (the plane subgraph of $G$ induced by the set of all uncrossed edges).

\begin{proposition}\label{claim: skeleton}
   If $G$ is an embedded graph such that each face of the skeleton $\mathrm{sk}(G)$ contains at most $q$ crossing points, then $\mu(G) \leq q + 1$.  
\end{proposition}

\begin{proof}
Let $F$ be a face of $\mathrm{sk}(G)$, and consider a crossing point $c$ inside $F$. As in the proof of \Cref{proposition: bounds on ribbon radius}, consider the boundaries $\mathcal{Z}(c,r)$ of balls $\mathcal{B}(c,r)$ for $r=0,1,2,\dots q$, and let $1 \leq r_c \leq q$ be the smallest integer such that $\mathcal{Z}(c,r)$ contains no dummy vertex of a crossing inside $F$. Then all edges on the face-boundary of $F$ must belong to $\mathcal{B}(c,r_c)$. Since $G$ is connected, the graph induced by edges on the boundary of $F$ and the crossed edges inside $F$ must be connected. Therefore, one can find a ribbon at $c$ within $\mathcal{B}(c,r_c)$. As $r_c \leq q$, we have $\mu(c) \leq q$. This implies that $\mu(G) \leq q+1$.
\end{proof}

An embedded graph $G$ is a \textit{$d$-framed graph} if $\mathrm{sk}(G)$ is simple, biconnected, spans all vertices  and all its faces have \emph{degree} at most $d$ (i.e., the face boundary has at most $d$ edges) \cite{bekos_dframe}. Examples of graphs that are $d$-framed include \textit{optimal $k$-plane graphs} for $k \in \{1,2,3\}$ \cite{bekos_optimal, schumacher1986struktur}, which are graphs with the maximum possible number of edges over all $k$-plane graphs. Another class of graphs related to $d$-framed graphs is \textit{$d$-map graphs}. A \textit{map graph} is the intersection graph of a map of \textit{nations}, where a nation is a region homeomorphic to a closed-disk, and the interiors of any two nations are disjoint \cite{chen2002map}. A \textit{$d$-map graph} has a map representation where at most $d$ nations intersect at a point. One can show that a graph $G$ is $d$-map if and only if it has an embedding where the set of all crossed edges occur as cliques inside faces of $\mathrm{sk}(G)$ with degree at most $d$ \cite{chen2002map}.  Since such a clique has an intersection-simple drawing with $O(d^4)$ crossings inside the face, it follows immediately from \Cref{claim:dframed} that $\mu(G)\in O(d^4)$ for $d$-map graphs. But, we can prove a better bound of $\mu(G)\in O(d^2)$, even
for a slightly more general class of graphs.

\begin{proposition}
\label{claim:dframed}
Let $G$ be an embedded graph with an intersection-simple drawing such that all crossed edges of $G$ are inside faces of $\mathrm{sk}(G)$ whose boundaries are  simple cycles of length at most $d$. Then $\mu(G) \leq \frac{d(2d+1)}{8}$.
\end{proposition}

\begin{proof}
Consider a crossed edge $e$ inside a face $F$. The number of edges that can cross $e$ is at most $(x-1)(|F|-x-1)$, where $|F|$ is the number of edges on $F$, and $x$ is the length of a shortest path along $F$ between the ends of $e$; this is maximized when $x= \left \lfloor |F|/2 \right \rfloor$. Since $|F| \leq d$, every edge is crossed by at most $\lceil \tfrac{d-2}{2}\rceil \lfloor \tfrac{d-2}{2} \rfloor$ other edges. Therefore, a $d$-map graph $G$ is $k$-planar for
$k \leq \tfrac{1}{4}(d{-}2)(d{-}1) = \tfrac{1}{4}(d^2-3d+2) \leq \tfrac{1}{4}(d^2-10)$. Since any crossing is inside a face of degree at most $d$, there is a path of length at most $d/4$ within the face boundary that connects some consecutive pair of endpoints of the crossing. Using \Cref{claim: ribbon radius in terms of ribbon length} (and the fact that the path of length $d/4$ consists of uncrossed edges), for any crossing point $c$, we have $\mu(c) \leq \max \left \{\frac{d^2-10}{4}, \frac{d^2-10}{4} + \frac{d}{8} -1 \right \} \leq \frac{d(2d+1)}{8} - 1$. Hence, $\mu(G) \leq \frac{d(2d+1)}{8}$. 
\end{proof}

In summary, the following classes of embedded graphs have small ribbon radius: $k$-plane graphs with a constant number of $\times$-crossings (\Cref{cor: bounded x crossings ribbon radius}), $k$-plane graphs with bounded-length ribbons at all crossings (\Cref{cor: small ribbon length}), graphs with a constant number of crossings (\Cref{cor: bounded crossings}), $d$-map graphs, $d$-framed graphs (\Cref{claim:dframed}), and more generally, graphs with a constant number of crossings in each face of its skeleton (\Cref{claim: skeleton}). 


\section{Ribbon Radius and Co-Separating Triples}\label{sec: co-separating triples}


In \Cref{theorem: vertex connectivity}, we present the main theoretical result of this paper. Briefly, the theorem shows that for any minimal vertex cut $S$ of a graph $G$ with $\mu := \mu(G)$, the graph $\Lambda(S, \mu)$ is connected and has diameter $O(\mu |S|)$. Moreover, $\Lambda(S, \mu)$ induces a co-separating triple $(A, X, B)$ of $(\Lambda(G), G)$, whose nucleus lies within $\Lambda(S, \mu)$, and the vertices of $\Lambda(G)$ contained within any hole of $\Lambda(S,\mu)$ only belong to one of the sets $A$ and $B$. (Recall from \Cref{section: ribbon radius} that a hole of $\Lambda(S,\mu)$ is a maximal region that does not contain any face of $\Lambda(S,\mu)$.) Since the nucleus of $(A,X,B)$ lies within $\Lambda(S,\mu)$, the nuclear-diameter is also in $O(\mu |S|)$; this will be crucial later when we design an algorithm to search for a minimum vertex cut. \Cref{fig: coseparating triple} provides an illustration for \Cref{theorem: vertex connectivity}.

\begin{figure}
    \centering
    \includegraphics[scale = 0.75]{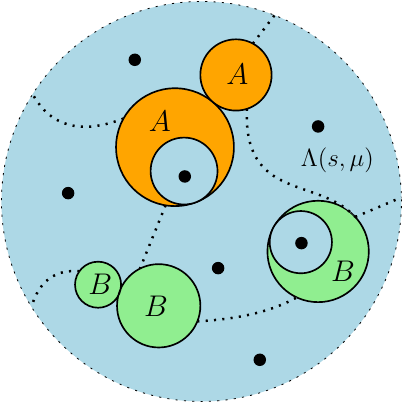}
    \caption{A figure for \Cref{theorem: vertex connectivity} depicting the structure of a co-separating triple of $(\Lambda(G),G)$. The graph is partitioned into different regions. The regions containing vertices of $S$ (shown with black dots) represent $\Lambda(s,\mu)$ for $s \in S$. The union of all such regions forms the graph $\Lambda(S,\mu)$. The remaining regions are holes of $\Lambda(S,\mu)$. These are bounded by solid edges representing the boundary $\mathcal{Z}(S,\mu)$. All vertices of $\Lambda(G)$ inside a hole belong entirely to $A$, or to $B$.}
    \label{fig: coseparating triple}
\end{figure}

\begin{theorem}\label{theorem: vertex connectivity}
Let $G$ be a connected embedded graph with a minimal vertex cut $S$. Let $\mu := \mu(G)$ be the ribbon radius of $G$. Then there exists a co-separating triple $(A,X,B)$ of $(\Lambda(G),G)$ where:
\begin{enumerate}
    \item For all edges $e \in E(G)$, all vertices of $e^\times$ belong entirely to $A \cup X$, or to $B \cup X$;\label{item: edges}
    \item $X \cap V(G) = S$ and $N_G[S] \subseteq \mathcal{B}(S,\mu)$;\label{item: x} 
    \item $X \subseteq \Lambda(S,\mu) \setminus \mathcal{Z}(S,\mu)$; \label{item: boundary}
    \item For any hole $F$ of $\Lambda(S,\mu)$, the vertices of $\Lambda(G)$ contained by $F$ either all belong to $A$, or they all belong to $B$; and \label{item: holes}  
   \item $\Lambda(S,\mu)$ is a connected graph of diameter $|S|(4\mu+1)$. \label{item: diameter}
\end{enumerate}
\end{theorem}

We devote the rest of this section to proving the theorem. We begin by defining the vertex sets of the co-separating triple through a suitable labelling of the vertices. Then, we show that this partition indeed forms a co-separating triple with all the necessary properties. There are three types of vertices in $\Lambda(G)$---vertices of $G$, dummy-vertices and face-vertices. Vertices of $G$ are labelled as follows: Fix an arbitrary flap $\phi_1$ of $S$. Each vertex in $S$ is labelled $X$, each vertex in $\phi_1$ is labelled $A$, and each vertex of $G$ in any flap other than $\phi_1$ is labelled $B$. Next, dummy-vertices are labelled as follows: For each dummy vertex $c$, if all the endpoints of the crossing at $c$ have the same label, then assign $c$ that label. In all other cases, label $c$ with $X$. Lastly, face-vertices are labelled as follows: For each face-vertex $f$, if all vertices of $V(G^\times)$ on the corresponding face boundary have the same label, then give $f$ that label. In all other cases, label it $X$.

To prove that $(A,X,B)$ is a co-separating triple, we verify all the three conditions in \Cref{def: co-sep}. Condition (C\ref{C1}) holds since $S$ is non-empty and defines at least two non-empty flaps. Condition (C\ref{C3}) is also easy to show. For any edge $e\in E(G \setminus (X \cap V(G))$, both its endpoints belong to the same flap and are labelled the same, say $A$, hence all dummy vertices on $e^\times$ can only be labelled $A$ or $X$. Condition~(\ref{C3}) implies Condition~(C\ref{C2}) for all edges of $G^\times$. Any other edge of $\Lambda(G)$ is incident with a face-vertex, and this can be labelled with $A$ or $B$ only if the other endpoint is labelled the same. So Condition~(C\ref{C2}) also holds and $(A,X,B)$ is a co-separating triple.

\begin{proof}[Proof of \Cref{theorem: vertex connectivity}(\ref{item: edges})]
Since $(A,X,B)$ is a co-separating triple of $(\Lambda(G),G)$, the statement already holds true for all edges $e \in E(G \setminus (X \cap V(G)))$. If $e$ is an edge incident with a vertex of $X \cap V(G)$, then all vertices on $e^\times$ are labelled $X$.
\end{proof}

\begin{proof}[Proof of \Cref{theorem: vertex connectivity}(\ref{item: x})]
Our labelling scheme ensures that $X \cap V(G) = S$; therefore, the nucleus of the co-separating triple is simply the closed neighbourhood $N_G[S]$ of $S$. We now show that $N_G[S] \subseteq \Lambda(S,\mu)$. Trivially, $S \subseteq \Lambda(S,\mu)$. For any vertex $v$ that is a neighbour of a vertex $s \in S$, we have from \Cref{obs: face distance between dummy-vertices of an edge} that $d_F(v,s)\leq \mu(G)$, so $v \in \Lambda(S,\mu)$. Therefore, $N_G[S] \subseteq \Lambda(S,\mu)$.    
\end{proof}

\begin{proof}[Proof of \Cref{theorem: vertex connectivity}(\ref{item: boundary})]
First, we look at vertices $v \in V(G^\times)$ labelled $X$, and show that $v\in \mathcal{B}(S,\mu(v))$. By $\mu(v)\leq \mu(G)-1$, all faces incident to $v$ will have face-distance at most $\mu(G)$ from $S$, implying that $v \notin \mathcal{Z}(S,\mu)$. If $v \in S$, then clearly $v\in \mathcal{B}(S,\mu(v))$, so assume that $v$ is a dummy vertex at the crossing point of two edges $\{e,e'\}$. If some endpoint of the crossing belongs to $S$, then by \Cref{obs: face distance between dummy-vertices of an edge}, $d_F(v,s)\leq \mu(v)$ for the endpoint $s \in S$; this implies that $s\in \mathcal{B}(v,\mu(v))$, and symmetrically, $v \in \mathcal{B}(s,\mu(v)) \subseteq \mathcal{B}(S,\mu(v))$. The remaining case is when all endpoints of the crossing are all labelled $A$ or $B$. As $v \in X$, all four endpoints cannot have the same label, and by \Cref{theorem: vertex connectivity}(\ref{item: edges}), the endpoints of each of $e$ and $e'$ are labelled the same. Hence, up to renaming $e$ and $e'$, assume that the endpoints of $e$ are labelled $A$ and those of $e'$ are labelled $B$. By definition of $\mu(v)$, there exists a ribbon $Q_v$ at $v$ such that $Q_v^\times \subseteq \mathcal{B}(v,\mu(v))$. Walking along the ribbon $Q_v$, we begin at a vertex labelled $A$, which is in flap $\phi_1$, and end at a vertex labelled $B$, which is in a flap other than $\phi_1$. Hence, the ribbon must contain a vertex $s \in S$. Therefore, $v \in \mathcal{B}(s,\mu(v)) \subseteq \mathcal{B}(S,\mu(v))$. 

Next, we look at vertices of $X$ that are face-vertices. If a vertex $u \in X$ is a face-vertex, then it belongs to a face of $G^\times$ where not all vertices are labelled $A$, or $B$. If the face contains vertices labelled both $A$ and $B$, then by walking along the face boundary, we must come across a vertex labelled $X$, since no two vertices labelled $A$ and $B$ can be adjacent (by Condition (C\ref{C2}) for co-separating triples). So $u$ is adjacent to a vertex $v \in V(G^\times)$ labelled $X$. Since $v \in \mathcal{B}(S,\mu(v))$, and $\mu(v) \leq \mu(G)-1$, we have that $u \in \Lambda(S,\mu)$. As $u$ is a face-vertex, trivially $u \notin \mathcal{Z}(S,\mu)$.
\end{proof}

\begin{proof}[Proof of \Cref{theorem: vertex connectivity}(\ref{item: holes})]
Every hole $F$ of $\Lambda(S,\mu)$ is bounded by a circuit $Z \subseteq \mathcal{Z}(\mathcal{B})$. By \Cref{theorem: vertex connectivity}(\ref{item: boundary}), no vertex of $X$ is in $Z$. Condition (C\ref{C2}) then implies that all vertices of $Z$ have the same label---without loss of generality, say it is $B$.   Since $F$ contains no part of $\Lambda(S,\mu)$ in the interior, \Cref{theorem: vertex connectivity}(\ref{item: boundary}) also implies that all vertices of $\Lambda(G)$ contained by $F$ can only have labels in $A \cup B$. It is sufficient to show that all vertices $v \in G^\times$ contained by $F$ are labelled $B$, since under our labelling scheme, this would imply that all face-vertices contained by $F$ are also labelled $B$. Suppose otherwise, for contradiction, that the hole contains a vertex $v \in A$. Then, either $v$ itself belongs to flap $\phi_1$, or $v$ is a dummy vertex and the endpoints of the crossing at $v$ belong to $\phi_1$. Now, choose a vertex $u \in \mathcal{B}(S,\mu)$ that belongs to $\phi_1$ (this exists by \Cref{theorem: vertex connectivity}(\ref{item: x})). Since $\phi_1$ is connected, one can find a $(u,v)$-path $P$ in $G^\times$ such that all vertices of $P$ are labelled $A \cup X$. Since $u$ and $v$ are on opposite sides of $Z$ in planar graph $\Lambda(G)$, path $P$ must intersect $Z$. This contradicts 
that all vertices of $Z$ are labelled $B$.      
\end{proof}

\begin{proof}[Proof of \Cref{theorem: vertex connectivity}(\ref{item: diameter})]
We will show that $\mathcal{B}(S,\mu)$ is connected, which will in turn imply that $\Lambda(S,\mu)$ is connected. By definition, the graph $\Lambda(s, \mu)$ is connected for each $s \in S$. It follows that $\mathcal{B}(s, \mu)$ is also connected, since any path in $\Lambda(s, \mu)$ that uses a face-vertex can be rerouted along the boundary of the corresponding face. This is valid because face boundaries in $G^\times$ are connected, as $G$ itself is connected. 

Suppose, for contradiction, that $\mathcal{B}(S,\mu)$ consists of at least two connected components $\mathcal{B}_1$ and $\mathcal{B}_2$.  As $\mathcal{B}(s,\mu)$ is connected for each $s \in S$, $\mathcal{B}_1 = \mathcal{B}(S_1,\mu)$ and $\mathcal{B}_2 = \mathcal{B}(S_2,\mu)$ for some non-empty disjoint sets $S_1 \subseteq S$ and $S_2 \subseteq S$. Since $\mathcal{B}_1$ and $\mathcal{B}_2$ are disjoint, $\mathcal{Z}(S_2,\mu)$ contains a circuit $Z$ that separates the embedding of $G^\times$ into two sides, one of which contains $\mathcal{B}_1$ and the other contains $\mathcal{B}_2$. As no vertex of $Z$ is labelled $X$ (\Cref{theorem: vertex connectivity}(\ref{item: boundary})), and by Condition (C\ref{C2}) of co-separating triples, all vertices of $Z$ must be labelled the same, say $A$.  Let $\phi_2$ be an arbitrary flap different from $\phi_1$, and choose a pair of vertices $u \in \mathcal{B}_1$ and $v \in \mathcal{B}_2$ such that $u,v\in \phi_2$. These must exist, since for any $s \in S$, there is a vertex of each flap adjacent to $s$ (by minimality of $S$), and by Observation~\ref{obs: face distance between dummy-vertices of an edge}, all neighbours of $s$ belong to $\Lambda(s,\mu)$.  As $\phi_2$ is connected, one can pick a $(u,v)$-path $P$ in $G$ such that all vertices of $P$ belong to $\phi_2$. Then the labels of all vertices of $P^\times$ belong to $B \cup X$. Since $P^\times$ connects a vertex of $\mathcal{B}_1$ with a vertex of $\mathcal{B}_2$, the path must contain a vertex of $Z$, contradicting that vertices of $Z$ are labelled $A$. So $\mathcal{B}(S,r)$ is connected, and hence $\Lambda(S,r)$ is connected.

We now show that the diameter of $\Lambda(S,\mu)$ is at most $|S|(4\mu +1)$. Since $\Lambda(S,\mu)$ is connected, for any pair of vertices $u,v \in \Lambda(S,\mu)$, there exists a $(u,v)$-path $P:= \langle u = w_1, w_2, \dots, w_{|P|} = v \rangle$ within $\Lambda(S,\mu)$. Now, we expand this into a walk $W$ by adding detours to vertices in $S$. For any vertex $w_i$, let $S(w_i)$ be the vertex in $S$ closest to $w_i$ in the graph $\Lambda(S,\mu)$. Let $W := w_1 \leadsto S(w_1) \leadsto w_1, w_2 \leadsto S(w_2) \leadsto w_2 \dots w_{|P|} \leadsto S(w_{|P|}) \leadsto w_{|P|}$. Next, for as long as $S(w_i) = S(w_j)$ for some $i \neq j$, prune the entire sub-walk between $S(w_i)$ and $S(w_j)$. We therefore end with a walk
$W^*$ where every vertex of $S$ occurs at most once. 
It consists of at most $2|S|$ paths between $w_i$ and $S(w_i)$, each of which has length at most $2\mu$ (since $d_F(w_i,S(w_i)) \leq \mu$ and the graph-distance in $\Lambda(G)$ is at most twice the face-distance), 
plus at most $|S|-1$ edges connecting $w_{i}$ to $w_{i+1}$. Therefore $|W^*|\leq (2|S|)(2\mu) + |S|-1 \leq |S|(4\mu+1)$ as required.
\end{proof}


\section{Connectivity Testing for Graphs with Small Ribbon Radius}\label{sec: finding min vertex and edge cut}

Our objective in this section is to prove the following:

\begin{theorem}\label{theorem: main theorem}
Given an embedded graph $G$ and an upper bound $\mu$ on the ribbon radius of $G$, one can compute its vertex connectivity $\kappa$ and its edge-connectivity $\lambda$ in time $2^{O(\mu\kappa)}|V(G^\times)|$ and $2^{O(\mu\lambda)}|V(G^\times)|$ respectively.
\end{theorem}

The overall approach of the proof follows the same idea as in \cite{biedl2022computing}. However, since the structure of co-separating triples used there differs from ours, and the graphs are not necessarily 1-plane, we cannot directly apply their algorithm and must instead adapt their techniques to suit our setting. 

As pointed out in \cite{biedl2022computing}, if $(A,X,B)$ is a co-separating triple of $(\Lambda(G),G)$, then $X \cap V(G)$ is a vertex cut of $G$ that separates $A \cap V(G)$ and $B \cap V(G)$. Conversely, \Cref{theorem: vertex connectivity}(2) shows that if $S$ is any minimal vertex cut of $G$, then there is a co-separating triple $(A,X,B)$ of $(\Lambda(G),G)$ such that $X \cap V(G) = S$. Therefore,  computing the vertex connectivity of an embedded graph is equivalent to finding a co-separating triple $(A,X,B)$ of $(\Lambda(G), G)$ with $|X \cap V(G)|$ minimized. 
To simplify 
this search, we make use of the fact that co-separating triples associated with minimal vertex cuts have low nuclear-diameter (\Cref{theorem: vertex connectivity}(5)). Taking the approach of \cite{biedl2022computing}, we break down $\Lambda(G)$ into low-diameter planar graphs (depending on $\mu$) $\Lambda_0, \Lambda_1, \dots, \Lambda_d$, for some finite integer $d$, such that the kernel resides completely in some such subgraph $\Lambda_i$. Since planar graphs with small diameter have small treewidth, we use the low-width tree-decompositions of $\Lambda_i$ to search for a co-separating triple of $(\Lambda_i,G)$ that can then be lifted into a co-separating triple of $(\Lambda(G), G)$. 

\subsection{Graphs \texorpdfstring{$\Lambda_i$}{Lambdai}}

Our description of the algorithm will focus on computing vertex connectivity. Only minor modifications to the algorithm are required to compute edge connectivity; these will be addressed later in \Cref{subsection: edge connectivity}. The input to our algorithm is an embedded graph $G$ (given via a rotation system of $G^\times$) and an upper bound $\mu$ on the ribbon radius. We first compute $\Lambda(G)$ and perform a Breadth-First Search (BFS) on it, thereby partitioning $V(\Lambda(G))$ into layers $V_0, V_1, \dots, V_d$, where vertices of layer $V_i$ are at distance exactly $i$ from the source vertex of the BFS. Let $s \in \{1,2,\dots, \kappa(G)\}$; we wish to test whether there exists a vertex cut
of size $s$. (We will test all values of $s$ until we succeed.) Set $w := s(4\mu+1)$; this corresponds to the bound on the nuclear diameter from Theorem \ref{theorem: vertex connectivity}(5). For $a \leq b$, let $\Lambda[V_a \dots V_b]$ denote the subgraph of $\Lambda(G)$ induced by vertices in the layers $V_a \cup \dots \cup V_b$ (for convenience, we define $V_i = \emptyset$ for all $i < 0$ and $i > d$).   

\begin{lemma}[Lemma 18 in \cite{arxivFull}]
\label{claim:U and L}
For $i\in \{0,\dots,d\}$, 
there exists a graph $\Lambda_i := \Lambda[V_{i-1} \cup \dots \cup V_{i+w+1}] \cup (V_{i-1},U_{i-1}) \cup (V_{i+w+1},L_{i+w+1})$, where $(V_{i-1},U_{i-1})$ and $(V_{i+w+1},L_{i+w+1})$ are graphs such that:
\begin{enumerate}
\item For any $u,v \in V_{i-1}$, there is a $(u,v)$-path in $(V_{i-1}, U_{i-1})$ if and only if there is a $(u,v)$-path in $\Lambda[V_{0}\cup \dots \cup V_{i-1}]$.   
\item For any $u,v \in V_{i+w+1}$, there is a $(u,v)$-path in $(V_{i+w+1}, L_{i+w+1})$ if and only if there is a $(u,v)$-path in $\Lambda[V_{i+w+1}\cup \dots \cup V_{d}]$.    
\item $\Lambda_i$ is planar and has radius at most $w{+}2$.
\end{enumerate}
 The total time to compute edge sets $\{U_i,L_i\}$ for all $i \in \{0,\dots,d\}$ is $O(|V(\Lambda(G))|)$.
\end{lemma}

The graphs $(V_{i-1},U_{i-1})$ and $(V_{i+w+1},L_{i+w+1})$ represent the projection of components of $\Lambda[V_0 \dots V_{i-1}]$ and $\Lambda[V_{i+w+1} \dots V_{d}]$ onto $V_{i-1}$ and $V_{i+w+1}$ respectively (see Figure \ref{fig: bfs tree} for an illustration and \cite{arxivFull,biedl2022computing} for more details). In \cite{arxivFull,biedl2022computing} it was shown, by way of \Cref{lem: Lambda to Lambda_i,lem: Lambda_i to Lambda}, that there exists a bijection between co-separating triples of $(\Lambda(G),G)$ and co-separating triples of $(\Lambda_i,G)$. One can easily verify that the proof in \cite{arxivFull} does not use 1-planarity, i.e., it works for all embedded graphs.

\begin{lemma}[Lemma 20 in \cite{arxivFull}]\label{lem: Lambda to Lambda_i}
    If there exists a co-separating triple $(A,X,B)$ of $(\Lambda(G), G)$ with nuclear diameter at most $w$, then there exists an index $i$ and a co-separating triple $(A_i,X,B_i)$ of $(\Lambda_i, G)$ where $A_i \subseteq A$, $B_i \subseteq B$, and $X \subseteq V_{i} \cup \dots \cup V_{i+w}$.
\end{lemma}

\begin{lemma}[Lemma 21 in \cite{arxivFull}]\label{lem: Lambda_i to Lambda}
   If $(A_i,X,B_i)$ is a co-separating triple of $(\Lambda_i, G)$ where $X \subseteq V_{i} \cup \dots \cup V_{i+w}$, then there exist sets $A$ and $B$ such that $A_i \subseteq A$, $B_i \subseteq B$, and $(A,X,B)$ is a co-separating triple of $(\Lambda(G), G)$. 
\end{lemma}

The broad idea of our algorithm is to iterate through $s=1,2,\dots$, where for each iteration, we test each $\Lambda_i$, in the order $i = 1,2,\dots,d$, for a co-separating triple $(A_i,X,B_i)$ where $X \subseteq V_i \cup \dots \cup V_{i+w}$ and $|X \cap V(G)| = s$. If we fail to find such a co-separating triple for all $\Lambda_i$ for the current value of $s$ we try $s + 1$. Thus, the problem now reduces to designing an algorithm that can test whether $\Lambda_i$ has a co-separating-triple $(A_i, X, B_i)$ with $X \subseteq V_i \cup \dots \cup V_{i+w}$ and $|X \cap V(G)| = s$.

\begin{figure}
    \centering
      \includegraphics[page = 1, scale = 0.7]{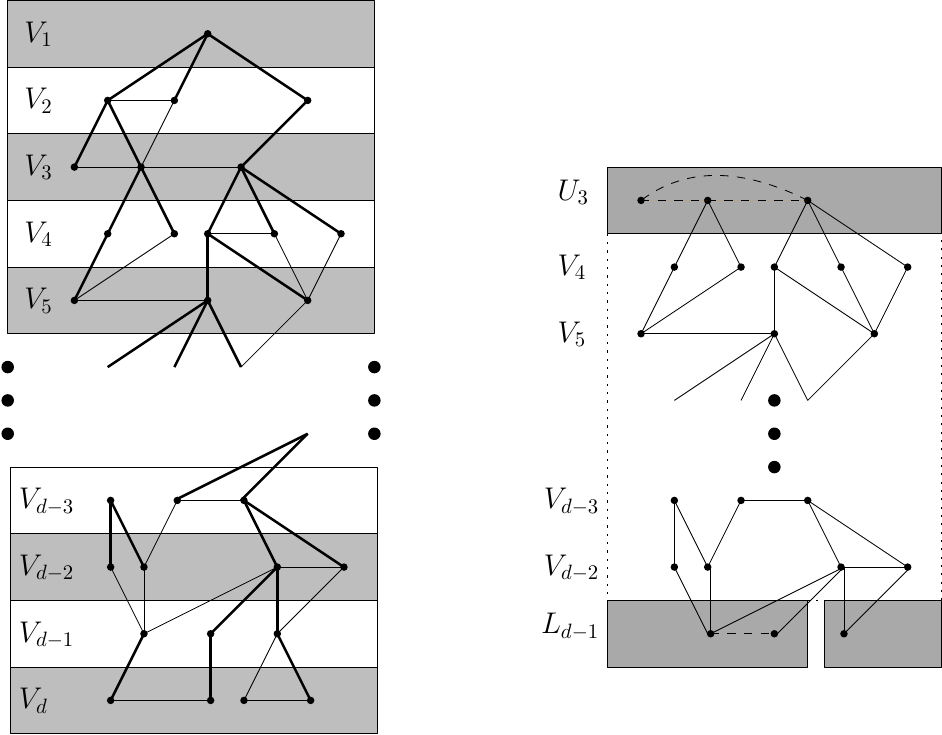}
      \caption{Constructing $\Lambda_i$ for $i=4$}
      \label{fig: bfs tree}
\end{figure}

We know from \cite{biedl2022computing} that $\Lambda_i$ is a planar graph of radius at most $w+2$, and hence we can find a tree decomposition (defined in \Cref{subsec: tree decompositions})  of $\Lambda_i$ of width $O(w)$ in $O(w|V(\Lambda_i)|)$ time \cite{baker1994approximation, eppstein2002subgraph}. 
It was shown in \cite{biedl2022computing,arxivFull} that one can dispense with testing Condition (C\ref{C3}) for co-separating triples by building an auxiliary graph $\Lambda_i^+$ that has the same vertex set as $\Lambda_i$ and treewidth at most constant times that of $\Lambda_i$. Unfortunately, this proof heavily depends on the 1-planarity of the graph. Therefore, in \Cref{app: sec: dp}, we provide a separate proof that applies to all embedded graphs with bounded ribbon radius. The next step in \cite{biedl2022computing,arxivFull} was to show that Conditions (C\ref{C1}) and (C\ref{C2}) can be phrased as a formula $\Phi$ in monadic second-order logic (MSOL). By Courcelle's theorem \cite{courcelle1990monadic}, we can therefore test in $O(f(w,|\Phi|)\, |V(\Lambda_i)|)$ time whether there exists a partition $(A_i,X,B_i)$ of $V(\Lambda_i^+)$ that satisfies Conditions (C\ref{C1}) and (C\ref{C2}). Unfortunately, function $f(\cdot)$ could be large, e.g.~a tower of exponentials.  
To improve the run-time, we give, in \Cref{subsec: dp}, a direct dynamic-programming-based algorithm whose running time is $2^{O(w)} |V(G^\times)| = 2^{O(\mu s)}|V(G^\times)|$.

\subsection{Tree Decompositions}\label{subsec: tree decompositions}

First, we review the definitions of tree decomposition and treewidth. 

\begin{definition}[Tree Decomposition]
A \emph{tree decomposition} of a graph $G$ is a tree $\calT$ whose
nodes are associated with subsets of $V(G)$ called \emph{bags} such that
\begin{enumerate}
    \item Every vertex $v$ of $G$ belongs to at least one bag of $\calT$;
    \item For every edge $(u,v)$ of $G$, there exists a bag $Y\in V(\calT)$ with $u,v\in Y$;
    \item For every vertex $v$ of $G$, the set of bags containing $v$ forms a connected subtree of $\calT$.
\end{enumerate}
\end{definition}

The \emph{width} of a tree decomposition is $\max_{Y\in V(\calT)} |Y|-1$,
and the \emph{treewidth} of a graph $G$ is the smallest width of a tree
decomposition of $G$. For purposes of dynamic programming, it is often helpful to have tree decompositions in a special form, which we call \emph{rooted binary tree decompositions}. (This is a weaker notion of the so-called nice tree decompositions; see \cite{bodlaneder_nice, kloks1994treewidth} for a reference.)

\begin{definition}[Rooted Binary Tree Decomposition]
A tree decomposition $\calT$ is a rooted binary tree decomposition if:
\begin{enumerate}
    \item $\calT$ is a rooted binary tree;
    \item If a node $\mathcal{N}$ has two children $\mathcal{N}_1$ and $\mathcal{N}_2$, then the bags at $\mathcal{N}$, $\mathcal{N}_1$ and $\mathcal{N}_2$ are all identical.
\end{enumerate}
\end{definition}

Any tree-decomposition $\calT'$ of width $\omega$ and $|\calT'|$ nodes can be transformed, in time $O(\omega|\calT'|)$, into a rooted binary tree decomposition $\calT$ of width $\omega$ and $O(|\calT'|)$ nodes. This is done as follows: First, root the tree $\calT'$ at an arbitrary node. Then, for every node $\mathcal{N}$ with $d \geq 2$ children in $\calT'$, replace $\mathcal{N}$ by a binary tree with $d$ leaves (hence $2d-1$ nodes), with every node of the binary tree having the same bag as $\mathcal{N}$. The children of $\mathcal{N}$ are now attached to the leaves of the binary tree. The resulting tree decomposition $\calT$ is a rooted binary tree decomposition of size $O(|\calT'|)$.

\subsection{Graphs \texorpdfstring{$\Lambda_i^+$}{Lambdaiplus}} \label{app: sec: dp}

Before giving the dynamic programming algorithm, we discuss the construction of graph $\Lambda_i^+$ so that we can forgo testing for Condition (C\ref{C3}). To understand the relevance of \Cref{lem: lambda_i+}(\ref{it:planarization_in_one}), it will be helpful to recall from \Cref{theorem: vertex connectivity} that there exists a co-separating triple of $(\Lambda(G),G)$ where all vertices of $e^\times$ belong entirely to $A \cup X$, or to $B \cup X$. When combined with \Cref{lem: Lambda to Lambda_i}, this implies that there exists a co-separating triple $(A_i,X,B_i)$ of $(\Lambda_i,G)$, for some $i \in \{0,\dots, d\}$, such that $X \subseteq V_{i} \cup \dots \cup V_{i+w}$, and all vertices of $e^\times \cap \Lambda_i$ belong entirely to $A_i \cup X$, or to $B_i \cup X$. (For \Cref{lem: lambda_i+}, it may be useful to recall that $w \in O(\mu s)$.) 

\begin{lemma}\label{lem: lambda_i+}
For every index $i \in \{0,\dots,d\}$, there exists a graph $\Lambda_i^+$, where $V(\Lambda_i^+) = V(\Lambda_i)$, such that the following hold:
\begin{enumerate}
    \item Let $(A_i,X,B_i)$ be a partition of $V(\Lambda_i)$ such that $X \subseteq V_{i} \cup \dots \cup V_{i+w}$.
    \begin{enumerate}
        \item \label{it:planarization_in_one} If $(A_i,X,B_i)$ is a co-separating triple of $(\Lambda_i,G)$ such that for any edge $e\in E(G)$ all vertices of $e^\times \cap \Lambda_i$ belong entirely to $A_i \cup X$, or to $B_i \cup X$, then $(A_i,X,B_i)$ satisfies Conditions (C\ref{C1}) and (C\ref{C2}) (in \Cref{def: co-sep}) for $\Lambda_i^+$.

        \item If $(A_i,X,B_i)$ satisfies Conditions (C\ref{C1}) and (C\ref{C2}) (in \Cref{def: co-sep}) for $\Lambda_i^+$, then $(A_i,X,B_i)$ is a co-separating triple of $(\Lambda_i,G)$.
    \end{enumerate}
       
    \item There is a rooted binary tree decomposition $\calT_i^+$ of $\Lambda_i^+$ of width $O(w)$ and $O(|V(\Lambda_i)|)$ nodes, and the total run-time to find $\calT_i^+$ for all $i \in \{0, \dots, d\}$ is $O(w^2|V(\Lambda(G))|)$.
\end{enumerate}
\end{lemma}

We prove the lemma in the remainder of this subsection.
Fix an index $i \in \{0,\dots, d\}$, and let $e \in E(G)$ be an edge such that $e^\times \cap \Lambda_i \neq \emptyset$. The graph $e^\times \cap \Lambda_i$ is a collection of vertex disjoint paths $\{e^\times(z_1 \dots z_2), e^\times(z_3 \dots z_4), \dots, e^\times(z_{2q-1} \dots z_{2q})\}$; we call the vertices $z_1,z_2\dots,z_{2q}$ \emph{transition-points} of $e^\times$ (\Cref{fig:transition-points}). To construct $\Lambda_i^+$, repeat the following process for all paths of $e^\times \cap \Lambda_i$, for all edges $e$ where $e^\times \cap \Lambda_i \neq \emptyset$: Make any dummy vertex on $e^\times(z_{2j-1}\dots z_{2j})$ adjacent to both transition-points $z_{2j-1}$ and $z_{2j}$, and add edge $(z_{2j-1},z_{2j})$ (\Cref{fig:transition-points}). 

\begin{figure}
\hspace*{\fill}
      \includegraphics[page = 2, scale = 0.75]{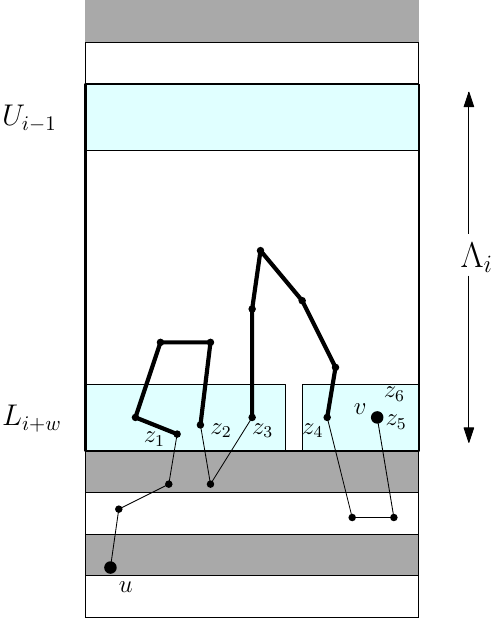}
\hspace*{\fill}
      \includegraphics[page = 3, scale = 0.75]{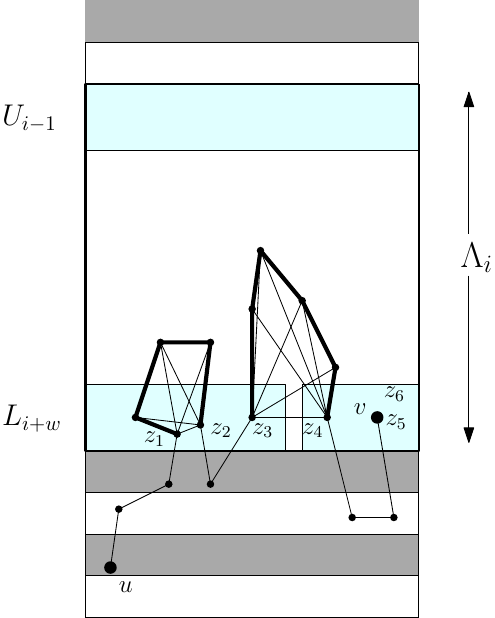}
\hspace*{\fill}
      \caption{An edge $e = (u,v)$; the paths of $e^\times$ in $\Lambda_i$ are bold.   We add edges from all vertices of $e^\times \cap \Lambda_i$ to the corresponding transition-points.} 
      \label{fig:transition-points}
\end{figure}

\begin{proof}[Proof of \Cref{lem: lambda_i+}(1)] 
It is straight-forward to prove \Cref{lem: lambda_i+}(1a). Let $(A_i,X,B_i)$ be a co-separating triple of $(\Lambda_i,G)$ as described. We need to show that these sets satisfy Conditions (C\ref{C1}) and (C\ref{C2}) for $\Lambda_i^+$. Condition (C\ref{C1}) holds trivially because $V(\Lambda_i^+) = V(\Lambda_i)$. Condition (C\ref{C2}) holds because each edge of $\Lambda_i^+$ is either an edge of $\Lambda_i$ or is an edge that connects two vertices of $e^\times \cap \Lambda_i$, for some $e \in E(G)$.

Next, we prove \Cref{lem: lambda_i+}(1b). Let $(A_i,X,B_i)$ be a partition of $V(\Lambda_i^+)$ as described. We need to show that $(A_i,X,B_i)$ is a co-separating triple of $(\Lambda_i,G)$. Condition (C\ref{C1}) holds trivially by $V(\Lambda_i^+) = V(\Lambda_i)$, and Condition (C\ref{C2}) holds by $E(\Lambda_i) \subseteq E(\Lambda_i^+)$. We are left with showing that Condition (C\ref{C3}) holds. Let $e \in E(G \setminus (X \cap V(G))$. First, consider the case where $e^\times \cap \Lambda_i = e^\times$. Here, there are only two transition points of $e^\times$---the two end vertices of $e$. Since the two end vertices are adjacent in $\Lambda_i^+$, we may assume without loss of generality that they belong to $A_i$. For every dummy vertex $z \in e^\times$, there is an edge in $\Lambda_i^+$ to both endpoints of $e$. Therefore, by Condition (C\ref{C2}), no vertex on $e^\times$ can belong to $B$. Next, consider the case where $e^\times \cap \Lambda_i \neq e^\times$, so there exist transition points $z_2,\dots,z_{2q-1}$. Each of these transition points must belong to $V_{i-1} \cup V_{i+w+1}$ (for otherwise, its neighbours on $e^\times$ would also be within $\Lambda_i$). In consequence, none of $z_2,\dots,z_{2q-1}$ can be in $X$. Up to symmetry, suppose that $z_2\in A_i$. We will show that all of $z_2,\dots,z_{2q-1}$ are in $A_i$. By induction, assume that $z_{2j} \in A_i$, for some $j<q$. Then $e^\times(z_{2j}\dots z_{2j+1}) \subseteq \Lambda[V_0\cup \dots \cup V_{i-1}]\cup \Lambda[V_{i+w+1}\cup \dots \cup V_{d}]$, hence $z_{2j+1}$ is in the same connected component of $(V_{i-1},U_{i-1}) \cup (V_{i+w+1},L_{i+w+1})$ as $z_{2j}$ (\Cref{claim:U and L}). By Condition (C\ref{C2}), all vertices in this component, and $z_{2j+1}$ in particular, belong to $A_i$. Since $e^\times(z_{2j+1} \dots z_{2j+2}) \subseteq \Lambda_i$, we have $(z_{2j+1},z_{2j+2}) \in E(\Lambda_i^+)$; Condition (C\ref{C2}) then implies that $z_{2j+2} \in A_i$. Therefore, by induction, $\{z_2,\dots,z_{2q-1}\} \subseteq A_i$. By this, every path in $e^\times \cap \Lambda_i$ has at least one end in $A_i$. Consider any such path $e^\times(z_{2j-1}\dots z_{2j})$. Since there is an edge in $\Lambda_i^+$ connecting every vertex on $e^\times(z_{2j-1}\dots z_{2j})$ to both $z_{2j-1}$ and $z_{2j}$, Condition (C\ref{C2}) forces all vertices on $e^\times(z_{2j-1}\dots z_{2j})$ to belong to $A_i \cup X$. Therefore, all vertices of $e^\times \cap \Lambda_i$ belong to $A_i \cup X$, and Condition (C\ref{C3}) holds.
\end{proof}

\begin{proof}[Proof of \Cref{lem: lambda_i+}(2)]
To construct a tree decomposition $\calT_i^+$ of $\Lambda_i^+$, we do some pre-processing on $\Lambda(G)$ so that we have a data structure $D(i)$, for each $i \in \{0,\dots,d\}$, that stores the following information: A list of dummy vertices in $\Lambda_i$, and for each dummy vertex $w$ at the crossing of two edges $\{e_1,e_2\}$, a corresponding set of at most 4 vertices for the transition points on the two sub-paths of $e_1^\times \cap \Lambda_i$ and $e_2^\times \cap \Lambda_i$ that intersect at $w$. To construct $D(i)$, first compute for every edge $e\in E(G)$ the values $(m_e,M_e)$ such that $e^\times\cap \Lambda_i$ is non-empty for all $m_e \leq i \leq M_e$. This takes $O(|e^\times|)$ time for each edge, so $O(|V(\Lambda(G))|)$ time for all edges in total. Then, bucket-sort the sets $\{m_e: e \in E(G)\}$ and $\{M_e: e \in E(G)\}$ in time $O(|E(G)|) \subseteq O(|E(\Lambda(G)|) = O(|V(\Lambda(G)|)$. For every $i \in \{0,\dots,d\}$, $e^\times \cap \Lambda_i \neq \emptyset$ if and only if $m_e \leq i+w+1$ and $M_e \geq i-1$. Since the lists $\{m_e: e \in E(G)\}$ and $\{M_e: e \in E(G)\}$ are sorted, one can compute, for all $i \in \{0,\dots,d\}$, the set $E_i := \{e: e^\times \cap \Lambda_i \neq \emptyset\}$. By traversing $e^\times$ for each $e \in E_i$, we can compute the set of paths in $e^\times \cap \Lambda_i$, and for each vertex in each path, store the two transition-points contained in that path. With this, we can create the data structure $D(i)$ as required. As any two vertices of $e^\times$ are within distance at most $2\mu < w$ in $\Lambda(G)$ (by \Cref{obs: face distance between dummy-vertices of an edge}), every edge $e \in E(G)$ is processed at most $w$ times in the computation of $E_i$ and construction of $D(i)$. Hence, the total time spent for all edges $e$ and indices $i$ is $O(w\sum_{e\in E}|e^\times|) = O(w|V(\Lambda(G))|)$.

From \Cref{claim:U and L}, the total to compute $\Lambda_i$ for all $i \in \{0, \dots, d\}$ is $O(|V(\Lambda(G))|)$. Since $\Lambda_i$ is planar and has radius at most $w+2$ \cite{biedl2022computing}, one can compute a tree decomposition $\calT_i$ of width $O(w)$ and $O(|V(\Lambda_i)|)$ nodes in $O(|V(\Lambda_i)|)$ time \cite{baker1994approximation,eppstein2002subgraph}. Having computed $\calT_i$, we build $\calT_i^+$ as follows.  Parse the tree decomposition $\calT_i$, and for every dummy-vertex $w$ that we see in a bag $X$, add the vertices that are stored in $D(i)$ associated with $w$ (there are at most four such vertices). It is easy to verify that this is a valid tree-decomposition; we leave the details of this to the reader. Since the size of each bag of $\calT_i$ is increased at most five-fold, 
the new tree decomposition has width $O(w)$.
Finally, this tree decomposition is turned into a rooted binary tree decomposition $\calT_i^+$
in $O(w|V(\Lambda_i)|)$ time (\Cref{subsec: tree decompositions}). 
Therefore, the total time to compute all $\calT_i^+$ is $\sum_{i=0}^d O(w|V(\Lambda_i)|)$ = $O(w^2|V(\Lambda(G))|)$, since each vertex of $\Lambda(G)$ belongs to at most $w+2$ graphs of $\Lambda_0,\dots,\Lambda_d$.
\end{proof}

\subsection{Dynamic Programming}\label{subsec: dp}

From \Cref{theorem: vertex connectivity}, we know that if $G$ has a minimal vertex cut of $s$, then there exists a co-separating triple $(A,X,B)$ of $(\Lambda(G),G)$ with nuclear diameter at most $w := 4s(\mu(G)+1)$ such that all vertices of $e^\times$ belong entirely to $A \cup X$, or to $B \cup X$. From \Cref{lem: Lambda to Lambda_i,lem: lambda_i+}, this co-separating triple translates to a partition $(A_i,X,B_i)$ of $\Lambda_i^+$, for some $i \in \{0,\dots,d\}$, such that $X \subseteq V_i \cup \dots \cup V_{i+w}$, and Conditions (C\ref{C1}) and (C\ref{C2}) are satisfied. Conversely, the existence of any such partition of $\Lambda_i^+$ implies a co-separating triple of $\Lambda_i$ (\Cref{lem: lambda_i+}), which in turn implies a co-separating triple $(A,X,B)$ of $\Lambda(G)$ (\Cref{lem: Lambda_i to Lambda}). The set $X \cap V(G)$ is then a vertex cut of $G$. Due to the chain of these equivalences, it is sufficient to test whether there exists a partition $(A_i,X,B_i)$ of $\Lambda_i^+$ such that $|X \cap V(G)| \leq s$, and $X \subseteq V_i \cup \dots \cup V_{i+w}$.  These conditions could easily be expressed in MSOL, and the test could therefore be done in linear time (with a large constant) using Courcelle's theorem \cite{arxivFull}. To improve the run-time, we will instead show here how to do this test by means of dynamic programming on the rooted binary tree decomposition $\calT_i^+$ of $\Lambda_i^+$.  

The main idea behind the dynamic programming algorithm is to proceed bottom-up in the tree, from the leaves to the root, and at each node, store all partial solutions corresponding to the subtree at that node, based on partial solutions that have been computed and stored at the children. We now define precisely what we mean by a partial solution. To simplify notations, we let $H: = \Lambda_i^+$. Let $\mathcal{N}$ be a node of $\mathcal{T}_i^+$, $V(\mathcal{N})$ be the vertices of $H$ in the bag corresponding to $\mathcal{N}$, and $H_\mathcal{N}$ be the subgraph of $H$ induced by the bags corresponding to the nodes in the subtree rooted at $\mathcal{N}$. A \textit{partial solution} $\mathcal{P}$ at the node $\mathcal{N}$ stores the following information.

\begin{enumerate}
    \item A function $f_\mathcal{P}: V(\mathcal{N}) \mapsto \{A, X, B\}$.  
    
    This function must satisfy Condition (C\ref{C2}) for $H[V(\mathcal{N})]$: for any two vertices $u, v \in V(\mathcal{N})$, if $uv \in E(H)$, then either $f_{\mathcal{P}}(u) \neq A$ or $f_{\mathcal{P}}(v) \neq B$. It also must satisfy that $f_{\mathcal{P}}(u)\neq X$ for $u \in V_{i-1} \cup V_{i+w+1}$.

    \smallskip
    (Note that while $A, X, B$ were previously used to denote sets of vertices, they are now used as labels on individual vertices to indicate their set membership. Also, for simplicity, we use the labels $A$ and $B$ instead of $A_i$ and $B_i$.)  
    
    \item An integer variable $\Sigma_\mathcal{P}$.
    
    \item Two boolean variables $\chi^A_\mathcal{P}$ and $\chi^B_\mathcal{P}$.
\end{enumerate} 

The variable $\Sigma_\mathcal{P}$ is intended to store the value of $|X \cap V(G)|$ in $V(H_\mathcal{N})$, while $\chi^A_\mathcal{P}$ and $\chi^B_\mathcal{P}$ are meant as indicators that are set to true if and only if $A_i \cap V(G)$ and $B_i \cap V(G)$ are non-empty in $H_\mathcal{N}$. For a partial solution at a node to reflect this information properly, they must be compatible with partial solutions at their children, which in turn must be compatible with partial solutions at their children, and so on, so that there is a non-conflicting assignment of vertices in $V(H_\mathcal{N})$ to $A_i$, $X$ or $B_i$. For this, we inductively define what are called \textit{valid partial solutions}.

A partial solution $\mathcal{P}$ at a leaf node $\mathcal{N}$ is valid if the following holds: 
\begin{itemize}
    \item $\Sigma_\mathcal{P} = |u \in V(\mathcal{N}): u \in V(G) \text{ and } f_\mathcal{P}(u) = X|$ and $\Sigma_\mathcal{P} \leq s$.
    
    \item For label $\Delta \in \{A,B\}$,  $\chi^\Delta_\mathcal{P} = 1$ if and only if there exists a vertex $u \in V(\mathcal{N}) \cap V(G)$ such that $f_{\mathcal{P}}(u) = \Delta$. 
\end{itemize}

A partial solution $\mathcal{P}$ at a non-leaf node $\mathcal{N}$ with a single child $\mathcal{N}_1$ is valid if there exists a valid partial solution $\mathcal{P}_1$ at $\mathcal{N}_1$ such that:

\begin{itemize}
    \item For any vertex $u \in V(\mathcal{N}_1) \cap V(\mathcal{N})$ we have $f_{\mathcal{P}_1}(u) = f_\mathcal{P}(u)$.
    
    \item $\Sigma_\mathcal{P} = \Sigma_{\mathcal{P}_1} +  |u \in V(\mathcal{N}) \cap V(G): f_{\mathcal{P}}(u) = X \text{ and } u \notin V(\mathcal{N}_1)|$ and $\Sigma_\mathcal{P} \leq s$.

    \item For label $\Delta \in \{A,B\}$, $\chi^\Delta_\mathcal{P} = 1$ if and only if either $\chi^\Delta_{\mathcal{P}_1} = 1$ or there exists a vertex $u \in V(\mathcal{N}) \cap V(G)$ such that $f_{\mathcal{P}}(u) = \Delta$.
\end{itemize}

A partial solution $\mathcal{P}$ at a non-leaf node $\mathcal{N}$ with two children $\mathcal{N}_1$ and $\mathcal{N}_2$ is valid if 
there exist valid partial solutions $\mathcal{P}_1$ and $\mathcal{P}_2$ at its children $\mathcal{N}_1$ and $\mathcal{N}_2$ such that (recall that the bags at $\mathcal{N}$, $\mathcal{N}_1$ and $\mathcal{N}_2$ are identical):
\begin{itemize}
    \item $f_\mathcal{P} = f_{\mathcal{P}_1} = f_{\mathcal{P}_2}$
    
    \item $\Sigma_\mathcal{P} = \Sigma_{\mathcal{P}_1} + \Sigma_{\mathcal{P}_2} - |u \in V(\mathcal{N}) \cap V(G): f_{\mathcal{P}}(u) = X|$ and $\Sigma_\mathcal{P} \leq s$.

    \item For $\Delta \in \{A,B\}$, $\chi^\Delta_\mathcal{P} = 1$ if and only if
    $\chi^\Delta_{\mathcal{P}_1} = 1$ or $\chi^\Delta_{\mathcal{P}_2} = 1$.
\end{itemize}

\medskip
The dynamic programming succeeds if there is a valid partial solution $\mathcal{P}$ at the root node such that $\Sigma_\mathcal{P} \leq s$, and both $\chi^A_\mathcal{P} = 1$ and $\chi^B_\mathcal{P} = 1$. If the dynamic programming succeeds, then the sets $(A_i, X, B_i)$ can be obtained by back-tracking through the computation that produced the solution.

\subsection{Testing Edge Connectivity}\label{subsection: edge connectivity}

We only need a few modifications in our approach to compute edge connectivity. First, construct a graph $\tilde{G}$ by subdividing every edge of $G$ exactly once.   In terms of the drawing of $G$,  we place the subdivision vertices close to an endpoint of the edge, before the first dummy vertex on the edge. The subdivision vertices are stored in a separate list $D(\tilde{G})$. For any set of edges $T \subseteq E(G)$, we use the notation $\tilde{T}$ to denote the set of subdivision vertices on edges of $T$.

\begin{proposition}\label{obs: edge-co-sep}
Let $T \subseteq E(G)$ be a set of edges. Then $T$ is an edge cut of $G$ if and only if there exists a co-separating triple $(A,X,B)$ of $(\Lambda(\tilde{G}), \tilde{G})$ such that $X \cap V(\tilde{G}) \subseteq \tilde{T}$.
\end{proposition}

\begin{proof}
If $T$ is an edge cut of $G$, then $\tilde{T}$ is a vertex cut of $\tilde{G}$. Applying \Cref{theorem: vertex connectivity} on a subset of $\tilde{T}$ that is a minimal vertex cut, we get a co-separating triple $(A,X,B)$ of $(\Lambda(\tilde{G}), \tilde{G})$ such that $X \cap V(\tilde{G}) \subseteq \tilde{T}$. For the other direction, suppose that there exists a co-separating triple $(A,X,B)$ of $(\Lambda(\tilde{G}), \tilde{G})$ such that $X \cap V(\tilde{G}) \subseteq \tilde{T}$. From \Cref{def: co-sep}, we know that $X \cap V(\tilde{G})$ separates $A \cap V(\tilde{G})$ and $B \cap V(\tilde{G})$. So it suffices to show that $A \cap V(G)$ and $B \cap V(G)$ are both non-empty. Suppose that $A$ contains a subdivision vertex $\tilde{e}$ of some edge $e=(u,v) \in E(G)$; without loss of generality, say it is placed close to $u$. This implies that $(u, \tilde{e})$ is an edge of $\tilde{G}$, and therefore, of $\Lambda(\tilde{G})$. By Condition (C\ref{C2}) of separating triples, this implies that $u \in A \cup X$. However, since $X \cap V(\tilde{G}) \subseteq D(\tilde{G})$, we must have $u\in A$. So, $A$ always contains a vertex of $G$, as does $B$.
\end{proof}

Therefore, computing edge connectivity is equivalent to the problem of finding a co-separating triple $(A,X,B)$ of $(\Lambda(\tilde{G}), \tilde{G})$ such that $X \cap V(\tilde{G}) \subseteq D(\tilde{G})$ and $|X \cap V(\tilde{G})|$ is the minimum possible. This can be achieved by only few minor changes to the algorithm for computing vertex connectivity: We let the input graph be $\tilde{G}$ instead of $G$ (it is easily seen that $\mu(G) = \mu(\tilde{G})$), and try to find a co-separating triple such that $X \cap V(\tilde{G}) \subseteq D(\tilde{G})$; this condition can be incorporated into the dynamic programming by enforcing the condition that $f_P(v) = X$ only if $v \in D(\tilde{G})$.

\subsection{Final Algorithm}

We now summarize the algorithm to test vertex and edge connectivity. Our input is a graph $G$ with an embedding described through a planar rotation system for $G^\times$. If our goal is to compute vertex connectivity, we simply compute $\Lambda(G)$. If our goal is to compute edge connectivity, we subdivide each edge of $G$ to obtain the graph $\tilde{G}$ (the subdivision vertices are placed close to an endpoint of the edge, as described before), and then compute $\Lambda(\tilde{G})$. We also assume that $\mu(G)$ is given as part of the input. Then, we do a BFS on $\Lambda(G)$ (or $\Lambda(\tilde{G})$ for edge connectivity), and obtain the layers $V_0, \dots, V_d$. Then, we compute the tree decompositions $\mathcal{T}_i^+$ for all $i \in \{0, \dots, d\}$. For each value of $s = 1,2,\dots,$ we test, by dynamic programming on $\calT_i^+$, for all $i \in \{0,\dots,d\}$, whether $\kappa(G) \leq s$ or $\lambda(G) \leq s$.

We now analyze the running time. As seen in \Cref{lem: lambda_i+}, the time to compute $\calT_i^+$ for all $i \in \{0,\dots,d\}$ is $O(w^2|V(\Lambda(G))|)$, where $w$ is the bound on nuclear diameter. However, this will be subsumed by the time to do dynamic programming, as we show next. We first bound the maximum possible number of valid partial solutions at each node. Each vertex of a bag has three choices of $A_i$, $X$ or $B_i$. Since each bag has $O(w)$ number of vertices, the number of different partitions of vertices of the node into $A_i$, $X$ and $B_i$ is $3^{O(w)}$. For each of these partitions, there are $s$ possible values of $\Sigma_\mathcal{P}$, and two possible values each for $\chi^A_\mathcal{P}$ and $\chi^B_\mathcal{P}$. Therefore, the total number of valid partial solutions is in $3^{O(w)}s \subseteq 2^{O(w)}$ since $w = O(\mu s)$. For each partial solution at a node, the time to verify if it is a valid partial solution is at most $O(w)$ times the size of the Cartesian product of valid partial solutions at its children---hence in $2^{O(w)} \times 2^{O(w)} \subseteq 2^{O(w)}$. Therefore, the time to compute and store all valid partial solutions at a node is also in $2^{O(w)}$. Since $\mathcal{T}_i^+$ has $O(|V(\Lambda_i)|)$ nodes (\Cref{lem: lambda_i+}), the time for dynamic programming is $2^{O(w)}|V(\Lambda_i)|$. When we sum across all of $\Lambda_i$, the total time is $\sum_{i = 1}^d 2^{O(w)}|V(\Lambda_i)| = 2^{O(w)}w|V(\Lambda(G))| \subseteq 2^{O(w)}|V(\Lambda(G))| = 2^{O(w)}|V(G^\times)|$. Since $w \in O(\mu s)$ and $s \leq \kappa(G)$ (or $s \leq \lambda(G)$), the runtime is in $2^{O(\mu \kappa)}|V(G^\times)|$ (or $2^{O(\mu \lambda)}|V(G^\times)|$). This completes the proof of \Cref{theorem: main theorem}. 


\section{Applications to Near-planar Graphs}\label{sec: algorithmic applications}

In \Cref{sec:applications}, we saw many classes of near-planar graphs with small ribbon radius. Now, we use \Cref{theorem: main theorem} to derive running times for computing vertex and edge connectivity for these classes of graphs (also refer to \Cref{table: results}). For many of these applications, we assume the drawing to be intersection-simple, which means that any pair of edges in the drawing intersect at most once, either at a common incident vertex, or at a crossing point. 

\begin{proposition}\label{cor: main result k-plane}
Given a $k$-plane graph $G$ with an intersection-simple drawing and ribbon radius $\mu$, we can compute $\kappa(G)$ and $\lambda(G)$ in time $2^{O(\mu\sqrt{k})}n$.   
\end{proposition}

\begin{proof}
  A simple $k$-plane graph has at most $4.108\sqrt{k}n$ edges \cite{pach1997graphs}. This implies that $\kappa(G) \leq \lambda(G) \leq \delta(G) \leq 8.216 \sqrt{k}$. Since each edge of $G$ is crossed at most $k$ times, $|V(G^\times)| \in O(nk^{3/2})$. Substituting these bounds into Theorem \ref{theorem: main theorem}, we get that both vertex and edge connectivity can be computed in time $2^{O(\mu\sqrt{k})}nk^{3/2} \subseteq 2^{O(\mu\sqrt{k})}n$.
\end{proof}

\begin{corollary}\label{cor: kplane constant mu}
Given a $k$-plane graph $G$, for some constant $k \in O(1)$, with an intersection-simple drawing such that $\mu(G) \in O(1)$, we can compute $\kappa(G)$ and $\lambda(G)$ in $O(n)$ time. 
\end{corollary}

\begin{corollary}\label{cor: q crossings per face}
Given a graph $G$ with an intersection-simple drawing such that each face of $\mathrm{sk}(G)$ has at most $q$ crossings, we can compute $\kappa(G)$ and $\lambda(G)$ in time $2^{O(q^{1.5})}n$.
\end{corollary}

\begin{proof}
Since $G$ has at most $q$ crossings, it is $q$-planar. By \Cref{claim: skeleton}, $\mu(G) \leq q+1$. Therefore, the time to compute $\kappa(G)$ and $\lambda(G)$ is in $2^{O(q^{1.5})}n$.
\end{proof}

\begin{corollary}\label{cor: dframed and dmap}
Given a graph $G$ with an intersection-simple drawing such that all crossings occur within faces of $\mathrm{sk}(G)$ whose boundary is a simple cycle with at most $d$ edges, we can compute $\kappa(G)$ and $\lambda(G)$ in time $2^{O(d^3)}n$.
\end{corollary}

\begin{proof}
Any edge of $G$ can be crossed only $O(d^2)$ times, and hence $G$ is $k$-plane for some $k \in O(d^2)$. We saw earlier in \Cref{claim:dframed} that $\mu(G) \in O(d^2)$. Therefore, by \Cref{theorem: main theorem}, $\kappa(G)$ and $\lambda(G)$ can be computed in time $2^{O(d^3)}n$.
\end{proof}

Note that \Cref{theorem: vertex connectivity,theorem: main theorem} do not stipulate that the embedded graph have an intersection-simple drawing. In \Cref{prop: crossing lemma}, we give an application for which the embedded graph need not have an intersection-simple drawing. 

\begin{proposition}\label{prop: crossing lemma}
Given an embedded graph $G$ with $o(\log n)$ crossings, we can compute $\kappa(G)$ and $\lambda(G)$ in time $n^{1+o(1)}$.
\end{proposition}

\begin{proof}
Let $n := |V(G)|$, $m := |E(G)|$ and $q := \mathrm{cr}(G)$, where $\mathrm{cr}(G)$ denotes the crossing number of $G$---the minimum number of crossing points over all drawings of $G$. By the well-known crossing number inequality \cite{crossing_lemma_1,crossing_lemma_2}, there is an absolute constant $c > 0$ such that $q \geq c\frac{m^3}{n^2}$ for all graphs with $m \geq 4n$. The best known constant till date is 1/29, due to Ackerman \cite{ackerman}, who showed that $q \geq \frac{m^3}{29 n^2}$ for all graphs with $m \geq 7n$. This implies that $\delta(G) \leq \frac{2m}{n} \leq \max\left \{14, \frac{2(29n^2 q)^{1/3}}{n} \right \} \leq \max \left \{14, 8(\frac{q}{n})^{1/3} \right \}$. Therefore, if $q \in o(\log n)$, both $\kappa(G), \lambda(G) \in O(1)$. By \Cref{cor: bounded crossings}, we have $\mu(G) \in o(\log n)$. Hence, by \Cref{theorem: main theorem}, both $\kappa(G)$ and $\lambda(G)$ can be computed in time $2^{o(\log n)}(n+o(\log n)) = n^{1+o(1)}$. 
\end{proof}


\section{Graphs with Large Ribbon Radius}\label{app :sec: large ribbon radius}

We have now seen a number of graph classes for which we can bound the ribbon radius. But what do graphs that have a large ribbon radius look like? Recall from \Cref{proposition: bounds on ribbon radius} that if a graph $G$ has at most $\gamma$ crossings with ribbon-radius exceeding $\alpha$, then $\mu(G)\leq \gamma+\alpha+1$. Therefore, if some graph $G$ has large ribbon-radius $\mu(G)$, then
setting $\alpha = \mu(G)/2 - 1$ in \Cref{proposition: bounds on ribbon radius}, we get that $G$ must have at least $\gamma \geq \mu(G)/2$ crossings with ribbon radius at least $\alpha + 1= \mu(G)/2$.
In fact, a closer inspection of the proof of \Cref{proposition: bounds on ribbon radius} reveals that for each crossing point $c$, and for all $1 \leq r \leq \mu(c)$, there is a vertex of $G^\times$ on the boundary of $\mathcal{B}(c,r)$ that has ribbon radius at least $\mu(c)-r$. This shows that crossings with large ribbon radius come clustered in the drawing.

Below, we give a construction of $k$-plane graphs with large ribbon radius. These graphs have an intersection-simple drawing and can be constructed to have connectivity $\Omega(\sqrt{k})$, which is asymptotically the best possible for $k$-plane graphs with intersection-simple drawings \cite{pach1997graphs}. A large value of connectivity ensures that there exists no simple algorithm that can find a minimum vertex or edge cut for these graphs. We construct these graphs $G$ such that for any minimum vertex cut $S$, the face-distance $d_F(s_1,s_2)$ is arbitrarily large in $\Lambda(G)$, for any two distinct vertices $s_1, s_2 \in S$. In other words, $\Lambda(S,r)$ is not a connected graph, for arbitrarily large values of $r$. Likewise, if $T$ is a minimum edge cut, we ensure that $d_F(\tilde{t}_1, \tilde{t}_2)$ is arbitrarily large in $\Lambda(\tilde{G})$, for any two distinct edges $t_1, t_2 \in T$. (Here, we follow the notations for edge connectivity as in \Cref{subsection: edge connectivity}.) Therefore, by \Cref{theorem: vertex connectivity}(\ref{item: diameter}), we infer that $\mu(G)$ must be large.

\begin{theorem}\label{thm: counterex edge conn}
For all positive integers $k,p,r$, where $p < 2\floor{\sqrt{k}}$, there exists a $k$-plane graph $G$ with a intersection-simple drawing such that:
\begin{enumerate}
    \item $\kappa(G) = \lambda(G) = p$;
    \item For any minimum vertex cut $S$ of $G$, $d_F(s_1,s_2) \geq r$ in $\Lambda(G)$ for all distinct $s_1,s_2 \in S$;
    \item For any minimum edge cut $T$ of $G$, $d_F(\tilde{t}_1, \tilde{t}_2) \geq r$ in $\Lambda(\tilde{G})$ for all distinct $t_1, t_2 \in T$ 
\end{enumerate}
\end{theorem}

\begin{figure}[h]
    \centering
    \includegraphics{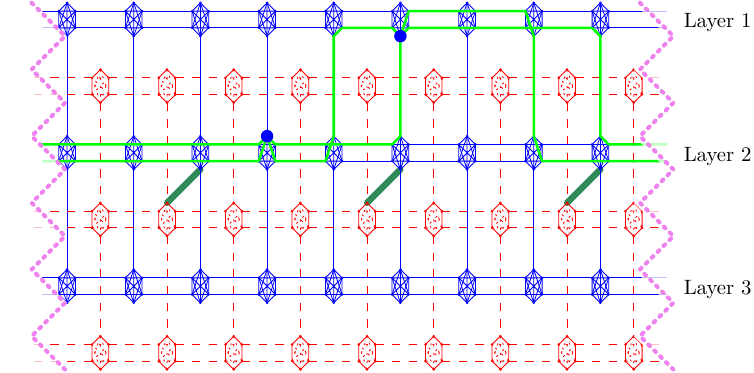}
    \caption{An illustration showing the graph constructed in the proof of \Cref{thm: counterex edge conn} for the values $k=4$, $p = 3$, and $r= 3$. Graph $R$ is dashed (red) while $B$ is solid (blue). The graph is drawn on a flat cylinder (the left and right boundaries are identified).}
    \label{fig: counterexample}
\end{figure}

\begin{proof}   
The idea is to construct two isomorphic $k$-plane graphs $R$ and $B$ (names correspond to colors used in Figure \ref{fig: counterexample}) that are $2\floor{\sqrt{k}}$-vertex-connected, and then to interleave the two disjoint copies, and connect them with a set of $p$ edges that are far apart from each other. The $p$ edges will form a unique minimum edge cut, and any set of $p$ vertices that cover these edges will form a minimum vertex cut. We describe this construction below in some detail. 

The graph $B$ (and likewise $R$) is constructed as follows. We begin with a clique on $2\lfloor \sqrt{k} \rfloor + 2$ vertices, drawn with the vertices in convex position and edges as straight-line segments. This yields an intersection-simple drawing in which each edge is crossed at most $\lfloor \sqrt{k} \rfloor \times \lfloor \sqrt{k} \rfloor \leq k$ times. We then take $\max\{2\floor{\sqrt{k}},rp\}$ copies of the clique arranged in a cyclical fashion and connect two consecutive cliques by a matching of $\floor{\sqrt{k}}$ edges. This forms a \textit{layer} of the drawing (Figure \ref{fig: counterexample}). We take $r$ layers, arrange them concentrically, and connect a clique of one layer with the copy of the same clique on a consecutive layer by a single edge. (This will therefore add $\max\{2\floor{\sqrt{k}},rp\}$ edges between two consecutive layers.) This completes the construction of $B$. One can argue that the graph is $2\floor{\sqrt{k}}$-vertex-connected using Menger's theorem \cite{Die12} which states that it suffices to show that any pair of non-adjacent vertices $u$ and $v$ can be connected by $2\floor{\sqrt{k}}$-vertex-disjoint paths. If $u$ and $v$ are within the same layer, the vertex-disjoint paths can be found by using the $2\floor{\sqrt{k}}$-matching edges that are incident on each clique. If $u$ and $v$ are in distinct layers $i$ and $j$,  the vertex-disjoint paths can be found using the $2\floor{\sqrt{k}}$-matching edges incident on each clique and $2\floor{\sqrt{k}}$ cliques on each layer to re-route the paths from layer $i$ to layer $j$; this is feasible since there are $\max\{2\floor{\sqrt{k}},rp\}$ cliques in each layer. (Figure \ref{fig: counterexample} illustrates this using thick light-green edges that connects the two fat blue vertices.)

Having constructed graphs $R$ and $B$, we interleave them such that there are $2r$ concentric layers in total, with a layer of $R$ alternating with a layer of $B$. Then we connect the two graphs by inserting $p$ edges (shown in fat dark-green in Figure \ref{fig: counterexample}) between the $\ceil{r/2}^\text{th}$ layer of $R$ and the $\ceil{r/2}^\text{th}$ layer of $B$. The $p$ edges are placed $r$ cliques apart; this is feasible by $\max\{2\floor{\sqrt{k}},rp\}$ cliques in each layer. This gives the graph $G$. As each edge is crossed at most $k$ times, $G$ is $k$-plane. Since each of $R$ and $B$ is $2\floor{\sqrt{k}}$-vertex-connected, each of them is also $2\floor{\sqrt{k}}$-edge-connected. As $p < 2\floor{\sqrt{k}}$, the $p$ edges form a unique minimum edge cut of $G$. Any set of $p$ vertices that covers all $p$ edges is a minimum vertex cut of $G$. It is now easy to verify that all three conditions in \Cref{thm: counterex edge conn} hold.  
\end{proof}

\begin{figure}
\centering
    \includegraphics[scale=0.75]{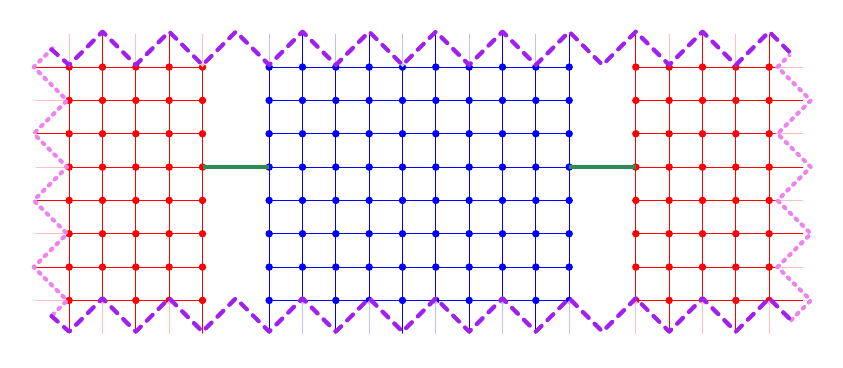}
    \caption{A toroidal graph obtained by connecting two grids by two edges on opposite sides of the grid. (The drawing is on the flat torus; identify the top with the bottom side and the left with the right side.)}
    \label{fig: torus}
\end{figure}


\section{Concluding Remarks}\label{sec: outlook}

In this paper, we showed that embedded graphs with small ribbon radius share the property with planar graphs that all vertices of a minimum vertex cut lie in a bounded diameter subgraph of $\Lambda(G)$. This property enabled a linear-time algorithm to test vertex and edge connectivity for many classes of near-planar graphs. One may ask whether our techniques extend to graphs drawn on other surfaces, such as a torus. We believe that this is unlikely. Consider the graph in Figure \ref{fig: torus} for an example. This graph is obtained by connecting two disjoint copies of planar grid graphs, each of which is 3-vertex-connected, with two edges on opposite sides of the grids. The two edges connecting the grids form a unique minimum edge cut, and contracting the two edges gives a unique minimum vertex cut. The face distance between endpoints of the two edges can be made arbitrarily large; so this example serves as an obstacle to generalizing our results to toroidal graphs. In another direction, one wonders what properties and algorithms of planar graphs can be generalized to embedded graphs with small ribbon radius. Of special interest are weighted versions: 
Can a minimum edge cut of an embedded edge-weighted graph with small ribbon radius be computed as fast as for planar graphs? 








\bibliography{bib2doi}

\end{document}